\documentclass{article}

\usepackage{comment}
\usepackage[utf8]{inputenc} % allow utf-8 input
\usepackage[T1]{fontenc}    % use 8-bit T1 fonts
\usepackage{lmodern}
\usepackage{hyperref}       % hyperlinks
\usepackage{url}            % simple URL typesetting
\usepackage{booktabs}       % professional-quality tables
\usepackage{amsfonts}       % blackboard math symbols
\usepackage{amsthm}
\usepackage{nicefrac}       % compact symbols for 1/2, etc.
\usepackage{microtype}      % microtypography
\usepackage{lipsum}
\usepackage{url,amsmath,setspace,amssymb,fullpage}
\usepackage{xcolor}
\usepackage{graphicx} % Allows including images
\usepackage{booktabs} % Allows the use of \toprule, \midrule and \bottomrule in tables
\usepackage{tikz}
\usetikzlibrary{shapes}
\usepackage{mathtools}
\usepackage{float}

\newtheorem{lemma}{Lemma}
\newtheorem{theorem}[lemma]{Theorem}
\newtheorem{corollary}[lemma]{Corollary}
\newtheorem{definition}{Definition}
\newtheorem{proposition}{Proposition}

\newtheorem{claim}{Claim}

\newcommand{\F}{\mathbb{F}}
\newcommand{\Fl}{\mathcal{F}(\ell)}
\newcommand{\IH}{I(\mathbb{H})}
\newcommand{\HM}{\mathbb{H}^{\ell}(\F_{q^2})}
\newcommand{\CH}{C^{\mathbb{H}}(\ell)}
\newcommand{\CA}{C^{\mathbb{A}}(\ell,2\ell)}
\newcommand{\X}{\mathbf{X}}
\newcommand{\Y}{\mathbf{Y}}
\date{}
\title{AFFINE HERMITIAN GRASSMANN CODES}

\author{

  \textbf{Fernando Pi\~nero Gonz\'alez} \\
  Department of Mathematics \\
  University of Puerto Rico in Ponce \\
  Ponce, PR. \\ 
  \texttt{fernando.pinero1@upr.edu} \\ %\thanks{Special thanks to Professor Jo Ellis-Monaghan for devotion to teaching us the beauty of graph polynomials.} \\

  \textbf{Doel Rivera Laboy}%\thanks{Special thanks to Professor Jo Ellis-Monaghan for devotion to teaching us the beauty of graph polynomials.} \\
  \\
  Department of Mathematics\\
  Pontifical Catholic University of Puerto Rico\\
  Ponce, PR \\
  \texttt{driveralaboy@pucpr.edu} \\
}

\begin{document}
\maketitle

\begin{abstract}
The Grassmannian is an important object in Algebraic Geometry. One of the many techniques used to study the Grassmannian is to build a vector space from its points in the projective embedding and study the properties of the resulting linear code.

We introduce a new class of linear codes, called Affine Hermitian Grassman Codes. These codes are the linear codes resulting from an affine part of the projection of the Polar Hermitian Grassmann codes. They combine Polar Hermitian Grassmann codes and Affine Grassmann codes. We will determine the parameters of these codes. %and discuss their minimum distance codewords.
\end{abstract}

% keywords can be removed
%\keywords{Coding Theory \and Grassmannian \and Linear Codes \and Grassmann Codes \and Affine Grassmann \and Polar Hermitian Grassmannian}

\section{Introduction}
Let $q$ be a prime power and $\F_q$ denote the finite field of $q$ elements. The Grassmannian, $\mathcal{G}_{\ell, m}$, is the collection of all subspaces of dimension $\ell$ of a vector space $V$ of length $m.$ We take $V=\F_q^m.$ The Grassmannian is a fascinating and well studied geometry with a rich algebraic structure. 

It is well known that the Grassmannian may be embedded into a projective space through the Pl\"ucker embedding. The embedding is performed as follows. For each vector space $W \in \mathcal{G}_{\ell, m}$, a matrix $M_W$ is taken such that the rowspace of $M_W$ is $W$. We then fix some ordering of the $\ell$-minors of an $\ell \times m$ matrix. Then $M_W$ is mapped to the projective point $p_W$, where the i-th position is the value of the i-th $\ell$-minor of $M_W$. In order to study the properties of the Grassmannian we make use of the Grassmann code. The Grassmann code is defined as the linear code generated by taking all of the points $p_W$ as the column vectors of a matrix. In \cite{Nogin}, Nogin studied the parameters of this linear code, denoted $C(\ell, m)$. 

Later, in \cite{BGT} Beelen, Ghorpade and H\o{}holdt introduced the linear code associated to one of the affine maps of the Pl\"ucker embedding of the Grassmannian. These codes are known as affine Grassmann codes, $C^\mathbb{A}(\ell,m)$ for $m \geq 2\ell$. Let $\delta = \ell (m-\ell)$, they proved the parameters of $C^\mathbb{A}(\ell,m)$ are:
$$[q^{\delta}, \binom{m}{\ell},  q^{\delta- \ell^2}\prod\limits_{i=0}^{\ell-1}q^\ell - q^i]$$
In that same paper, they studied the automorphisms of the code and counted the minimum distance codewords. Their duals, $C^\mathbb{A}(\ell,m)^\perp$, and related codes were studied in \cite{BGT2} and \cite{AfGrass}. Moreover, in \cite{AfGrass}, the minimum weight codewords of the dual code were classified and counted.

In \cite{Cardinali_2013}, Cardinali and Giuzzi introduced polar Grassmann codes, linear codes which are closely related to the Grassmann codes previously mentioned. The polar Grassmannian is a subvariety consisting of the subspaces of $V$ isotropic under a given form. Cardinali and Giuzzi determined the parameters of polar Grassmann codes under a symplectic form (polar symplectic Grassmannian)\cite{Symplectic}, under a orthogonal form (polar orthogonal Grassmannian)\cite{cardinali2014line}\cite{Cardinali_2018} and under a Hermitian form (polar Hermitian Grassmannian)\cite{Cardinali} for $\ell = 2$.

Much is known about both affine Grassmann codes and polar Grassmann codes. However, in this work we study the linear code associated to one of the affine maps of the polar Hermitian Grassmannian. We first determine the length, dimension and minimum distance of this code. We also determine some of  its automorphisms and its dual code. Furthermore, we characterize the minimum weight codewords of the dual code.

\section{Preliminaries} In this manuscript, we shall use the following notation. %The letter $q$ will denote a prime power of the prime $p$. We'll denote the finite field of $q$ elements by $\F_{q}$. 
For positive integers $m$ and $n$, an $m\times n$ matrix is a rectangular array of $m$ rows and $n$ columns of entries over $\F_{q}$ or $\F_{q^2}$. Matrices will be denoted by upper case letters such as $A, B, C, M, N$. Generic matrices will be denoted by $X$, $Y$ or $Z$. In this work, we will focust mostly on square matrices over $\F_{q}$ and $\F_{q^2}.$  

\begin{definition}
    Let $M$ be a square matrix. Suppose that $I$ is a subset of the rows of $M$ and $J$ is a subset of the columns of $M$. The minor $det_{I,J}(M)$ is the determinant of
the submatrix of $M$ obtained from the rows $I$ and columns $J$. In the case $I = \{ i\}$ and $J = \{ j \}$ we shall denote by $det_{I,J}(M)$ minor by the $(i,j)-th$ entry $M_{i,j}$.
\end{definition}
Example:
Let $M = \begin{bmatrix}
    1&0&0\\0&1&0\\0&0&1
    \end{bmatrix}$, $I = \{1,2\}$ and $J = \{2,3\}$. Then $$det_{I,J}(M) = \begin{vmatrix}
    0&0\\1&0
    \end{vmatrix} = 0.$$ 
    
As $\F_{q^2}$ is an extension of degree $2$ of the field $\F_q$, for $\alpha \in \F_{q^2}$ we denote the conjugate of $\alpha$, $\alpha^q$, by $\bar{\alpha} = \alpha^q$. Likewise, we may conjugate matrices by conjugating each entry. Therefore we denote the conjugate of $M$, the matrix $[M_{i,j}^q]$ by ${M}^{(q)}$. We recall the following definition:

\begin{definition}
    For a matrix $M \in \mathbb{M}^{n\times n}(\F_{q^2})$, we denote $M^*$ as the conjugate transpose of $M$, that is $M^* = [m_{ji}^q]$.
\end{definition}
Example:\\
Let $M = \begin{bmatrix}
    a&b\\c&d
    \end{bmatrix}$ 
then $M^* = \begin{bmatrix}
    a^q&c^q\\b^q&d^q
    \end{bmatrix}$
%\begin{definition}
 %   For $\alpha \in \F_{q^2}$ we define $\Bar{\alpha} = \alpha^q$. Furthermore for a matrix $M$, $\Bar{M} = [m_{ij}^q]$. Finally for a matrix $M$, we denote $M^*$ as the conjugate transpose of $M$, that is $M^* = [m_{ji}^q]$.
%\end{definition}
\begin{definition}
    A matrix over $\F_{q^2}$ is Hermitian if $M^* = M$.
\end{definition}
Example:\\
Let $q=2$, note $M = \begin{bmatrix}
    1&\alpha&1\\\alpha^2&1&0\\1&0&1
    \end{bmatrix}$ is Hermitian.
In general, a Hermitian matrix $H$ over $\F_{q^2}$ satisfies:
$$H = \begin{cases}
H_{i,j}\in \F_q & i=j\\
H_{i,j}\in \F_{q^2} & i>j\\
H_{i,j} = {H_{j,i}}^q & j<i\\
\end{cases}$$

\subsection{affine Hermitian Grassman code}
%For $\ell \geq 1$, %we define $X = [X_{ij}]$ as the $\ell \times \ell$ matrix of indeterminates over $\mathbb{F}_{q}$. We denote $\mathbb{F}_{q}[X]$ as the polynomial ring over $\mathbb{F}_{q}$ with $\ell^2$ indeterminates. Interchangeably
For $\ell \geq 1$ we define $X = [X_{ij}]$ as the $\ell \times \ell$ matrix of $\ell^2$ indeterminates over $\F_{q^2}$. As we are interested in Hermitian matrices, we shall define $X_{i,i}$ as an element of $\F_q$ (satisfying $X_{i,i}^q = X_{i,i})$. For $i<j$, $X_{i,j}$ is an element of $\F_{q^2}$ (satisfying $X_{i,j}^q = X_{i,j})$. Then $X_{j,i}$ for $i<j$ is the conjugate of its transposed position, that is, $X_{j,i} = X_{i,j}^q.$  In a similar manner, we define $Z = [Z_{i,j}]$. $Z$ is the generic $\ell \times \ell$ matrix of $\ell^2$ indeterminates over $\F_{q}$. %$Y$ is the generic $\ell \times \ell$ matrix of $\ell^2$ indeterminates over $\F_{q^2}$ such that $Z_{i,j}+Z_{i,j}^q = 0$.
%we define $X = [X_{ij}]$ as the $\ell \times \ell$ Hermitian matrix of $\binom{\ell+1}{2}$ indeterminates where $x_{ii} \in \mathbb{F}_{q}[X]$, $x_{ij} \in \mathbb{F}_{q^2}[X]$ and, for $j<i$, $y_{ij} = y_{ji}^q$. \\

%Let $F_q[x]$ be the Polynomial Ring of $\ell^2$ indeterminates over $\F_q$, we define the following:
%\begin{definition}
%$$X = \begin{cases}
%X_{i,j}\in \F_q & i=j\\
%X_{i,j}\in \F_{q^2} & i>j\\
%X_{i,j} = {X_{j,i}}^q & j<i\\
%\end{cases}$$
%$$Y = \begin{cases}
%Y_{i,j}\in \F_q &\\
%\end{cases}$$
%$$Z = \begin{cases}
%Z_{i,j}\in \F_{q^2} & \\
%Z_{i,j} = -{Z_{j,i}}^q & \\
%\end{cases}$$
%\end{definition}

\begin{definition}
    $\Delta(\ell)$ is the set of all minors of the matrix X. That is:
    $$\Delta(\ell) := \{ det_{I,J}(X), I,J \subseteq [\ell], \#I = \#J  \} $$
    
\end{definition}
We remark that $\# \Delta(\ell) = \binom{2\ell}{\ell}$.

Example: 
For $\ell = 2$, let the matrix $\X$ be equal to
  $$\X = \begin{bmatrix}
    X_1&X_2\\X_2^q&X_3
    \end{bmatrix},$$ 
Then $\Delta(\ell) = \{ 1, X_1, X_2, X_3, X_2^q, X_1X_3-X_2^{q+1} \} $.

\begin{definition}
We define $\Fl$ as the subspace of $\mathbb{F}_{q^2}$--linear combinations of elements of $\Delta(\ell)$.

$$\Fl := \{ \sum\limits_{I,J \subseteq [\ell], \# I = \# J} f_{I,J} det_{I,J}(X) | f_{I,J} \in \F_{q^2}   \} $$
\end{definition}
Example:\\
Let $\ell = 2$ then all elements of $\Fl$ are of the form:
$$f_{\emptyset, \emptyset}+f_{\{1\},\{1\}}X_1+ f_{\{1\},\{2\}}X_2+ f_{\{2\},\{2\}}X_3+ f_{\{2\},\{1\}}X_2^q+ f_{\{1,2\},\{1,2\}}(X_1X_3-X_2^{q+1})$$

%\begin{lemma}\cite{BGT}
%    The cardinality of $\Delta(\ell)$ is $\binom{2\ell}{\ell}$.
%\end{lemma}
%\begin{lemma}\cite{BGT}
%    $dim_{\F_{q^2}}(\mathcal{F}(\ell)) = \binom{2\ell}{\ell}$
%\end{lemma}
We remark that the minors $det_{I,J}(X)$ as polynomials on $X_{i,j}$, it is simple to determine that $\dim_ \Delta(\ell) = \binom{2\ell}{\ell}$. Also, in a similar manner to $\Fl$, $\Fl_Z$ may be defined as the subspace of $\mathbb{F}_{q}$--linear combinations of elements of $\Delta(\ell)$ when taking $Z$ as the generic matrix instead of $X$.

%\begin{definition}
%    Given $f = \sum_{M\in \Delta(\ell)} a_{M}M \in \Fl $ where $a_M \in \F_{q^2}$ for every $M \in \Delta(\ell)$, the support of f is the set
 %   \begin{center}
 %       $supp(f) := \{M \in \Delta(\ell) : a_M \neq 0\}$
 %   \end{center}
%\end{definition}
\begin{definition}
    $\mathbb{H}^{\ell}(\F_{q^2})$ denotes the space of all $\ell \times \ell$ Hermitian matrices with entries in $\F_{q^2}$. That is $$\mathbb{H}^{\ell}(\F_{q^2}) = \{H \in \mathbb{M}^{\ell \times \ell}(\mathbb{F}_{q^2})| M = M^* \}.$$
\end{definition}
%\begin{definition}
%    $\mathbb{SH}^{\ell}(\F_{q^2})$ denotes the space of all $\ell \times \ell$ skew-Hermitian matrices with entries in $\F_{q^2}$. That is $$\mathbb{SH}^{\ell}(\F_{q^2}) = \{H \in \mathbb{M}^{\ell \times \ell}(\mathbb{F}_{q^2})| M = -M^* \}.$$
%\end{definition}
Now we define, the evaluation of an element $f \in \Fl$ at any Hermitian matrix $P \in \mathbb{H}^{\ell}(\F_{q^2})$. For any $f \in \mathbb{F}_{q^2}[X]$ and $P \in \mathbb{H}^{\ell}(\F_{q^2})$  the evaluation $f(P) $ is obtained by replacing the variable $X_{i,j}$ by the element $P_{i,j}$. From now on, we shall denote $n$ by $n = q^{\ell^2} = \# \mathbb{H}^{\ell}(\F_{q^2}) $. For the evaluation map we fix an arbitrary enumeration $P_1, P_2, ..., P_n$ of $\mathbb{H}^{\ell}(\F_{q^2})$.

\begin{definition}
    The evaluation map of $\mathbb{F}_{q^2}[X]$ is the map
    \begin{center}
        $ev$ : $\mathbb{F}_{q^2}[X] \rightarrow \F_{q^2}^{n}$ defined by $Ev(f) := (f(P_1),...,f(P_{n}))$.
    \end{center}
\end{definition}

In order to prove the evaluation map is an injective function, we need the following lemma on the elements of $\Fl$.

\begin{definition} 
Denote by $\IH$ the ideal generated by $X_{i,j}^q - X_{j,i}$ and,  $X_{i,j}^{q^2} - X_{i,j}$ where $i<j$ and $X_{i,i}^q = X_{i,i}$ That is $$\IH := \langle X_{i,j}^q - X_{j,i},X_{i,j}^{q^2} - X_{i,j}, X_{i,i}^q-X_{i,i} \rangle.$$ 
\end{definition}
The set of Hermitian matrices $\mathbb{H}^{\ell}(\F_{q^2})$ is precisely the set of solutions to all polynomial equations generating $\IH.$ We shall now prove the following lemma relating $\Fl$ and $\IH$.
\begin{lemma}\label{lem:ideal}
 There are no elements in common between $\Fl$ and $\IH$ except $0.$
 
 That is $$\Fl \cap \IH = \{ 0 \}.$$
 
\end{lemma}
\begin{proof}

Let $f \in \Fl.$ Consider $f_H = f \mod \IH$, the reduction of $f$ modulo the ideal $\IH$. This may be attained simply by replacing the entry $X_{j,i}$ for $i<j$ with its conjugate transpose $X_{i,j}^q.$ This is given by the polynomial equation $X_{i,j}^q - X_{j,i} = 0.$

Any determinant will have at most degree $q+1$ on the variable $X_{i,j}$ for $i<j$ and degree at most $1$ on the variable $X_{i,i}$. This implies no single determinant is in $\IH$ as that would imply the degree in $X_{i,i}$ is at least $q$ and the degree in $X_{i,j}$ is at least $q^2$. Now we shall prove that even after setting $X_{j,i} = X_{i,j}^q$ all linear combinations in $\Fl$ are still independent.

Let $I = \{i_1 < i_2 < \cdots < i_m\}$ and $J = \{j_1 < j_2 < \cdots < j_m\}$. The term $\prod\limits_{s=1}^m X_{i_s,j_s}$ is a term appearing only in $det_{I,J}$. The reduction $\prod\limits_{s=1}^m X_{i_s,j_s} \mod \IH$ is obtained by replacing any $X_{i_s,j_s}$ by $X_{j_s,i_s}^q$ whenever $i_s > j_s$.  Denote this monomial by $Mn_{I,J}$.

Now let $I' = \{i'_1 < i'_2 < \cdots < i'_{m'}\}$ and $J' = \{j'_1 < j'_2 < \cdots < j'_{m'}\}$. We shall prove that if $I \neq I'$ or $J \neq J'$ then the monomial $Mn_{I,J} \neq Mn_{I',J'}$.

Clearly if the sizes of $I$ and $J$ are different from the size of $I'$ and $J'$ then  both $Mn_{I,J}$ and $Mn_{I',J'}$ will have a different number of monomials. We may assume $\# I = \# J = \# I' = \# J'.$ Suppose now that $I \neq I'$ or $J \neq J'$. We may assume there exists a pair $(i_s, j_s) \neq (i_s', j_s')$. If $\{ i_s, j_s  \} \neq \{ i_s', j_s'  \}$ then $Mn_{I,J}$ contains a different variable than $Mn_{I',J'}$. If $\{ i_s, j_s  \} = \{ i_s', j_s'  \}$, the condition $(i_s, j_s) \neq (i_s', j_s')$ implies $i_s' = j_s$ and $j_s = i_s'$. The reduction $\mod \IH$ will map $X_{j_s, i_s}$ to $X_{i_s, j_s}^q$, which implies $Mn_{I,J} \neq Mn_{I',J'}$.

As all reductions of all determinants modulo $\IH$ are independent, the statement follows. \end{proof}

\begin{lemma}\label{lem:inyective}
 The evaluation map $ev$ is injective.
\end{lemma}
\begin{proof}
Note that the ideal $\IH$ is precisely the ideal of polynomial functions which vanish on $\HM$. That is $ev(f) = 0$ if and only if $f \in \IH.$ As Lemma \ref{lem:ideal} implies there is no nonzero element in both $\Fl$ and $\IH$, we then have that $Ker(ev) = \Fl \cap \IH = \{0\}.$ Thus the evaluation map is injective. 

\end{proof}
%$\mathbb{H}$ may be described as the set of solutions to the following polynomial equations:
%\begin{equation}
%    {y_{ii}}^q-{y_{ii}}
%%\end{equation}
%\begin{equation}
%%    {y_{ij}}^{q^2}-{y_{ij}} , i<j
%\end{equation}
%The ideal of the functions that vanish in $\mathbb{H}$ is I($\mathbb{H}$) and it is generated by the previous equations.\\
%Note that for a linear combination of elements in $\Delta(\ell)$, ${y_{ii}} \in \F_q $ has degree at most 1 and ${y_{ij}} \in |F_{q^2}$ for $i<j$ has degree at most q+1.\\
%This implies the degree is strictly lower than the previous characteristic equations.\\
%This implies that the linear combination cannot be expressed as a nontrivial combination of the previous equations."\\
%Here goes the lemma where no linear combination of elements in $\Delta(\ell)$ may be written as a nontrivial combination of the previous equations. (Note this is done by looking at the degrees)\\
%This implies that for all $f\neq 0 \in \Delta(\ell)$, $f\notin I(\mathbb{H})$.\\ 
%By consequence, the only element which maps to 0 is $f = 0$.\\
%Therefore, the map is injective. \qed
%Then goes the lemma that no linear combination of elements in $\Delta(\ell)$ is in I(H). Therefore by consecuence, the only element which maps to 0 is the 0 function. By consequence the map is injective.\\

We are now ready to define affine Hermitian Grassmann codes.

\begin{definition}
    The affine Hermitian Grassmann code $C^{\mathbb{H}}(\ell)$ is the image of $\Fl$ under the evaluation map $ev$. That is
    
    $$ \CH := \{ (f(P_1), f(P_2), \ldots, f(P_n)) | f \in \Fl  \} $$
\end{definition}

We remark that if $A$ is skew-Hermitian, then $A = \alpha H$ where $H$ is Hermitian and $\alpha^q = -\alpha$. This implies $det_{I,J}(A) = \alpha^{\#I} det_{I,J}(H)$. Which, in turn, implies that all evaluations under $ev$ for skew-Hermitian matrices are scalar multiples of evaluations for $\mathbb{H}^{\ell}(\F_{q^2})$. Therefore they generate the same code, that is the affine Hermitian Grassmann code and the affine skew-Hermitian Grassmann code are equivalent.

As a direct consequence of the Rank--Nullity theorem and Lemma \ref{lem:inyective} we obtain the following corollary. 
\begin{corollary}
$$\dim \CH  = \binom{2\ell}{\ell}.$$

\end{corollary}

Now we compare our code $\CH$ to the affine Grassmann code $\CA$ introduced by Beelen, Ghorpade and H\o{}holdt. It turns out that both codes are very similar to each other. However some key differences between $\CH$ and $\CA$ remain. Now we state the definition of the code $\CA$.

\begin{definition}\cite{BGT}

    The affine Grassmann code $\CA$ is obtained by evaluating $f \in \Fl_Z$ onto all $\ell \times \ell$ matrices over $\F_q$.
    
    $$ \CA := \{ (f(P_1), f(P_2), \ldots, f(P_n)) | f \in \Fl_Z  \} $$

where $P_1, P_2, \ldots, P_n$ are the elements of $\mathbb{M}^{\ell \times \ell}(\F_q)$ in some order.
\end{definition}

We remark that $\# \mathbb{M}^{\ell \times \ell}(\F_q) = \# \HM = q^{\ell^2}.$ Thus both $\CA$ and $\CH$ have the same length. They also have the same dimension, $\binom{2\ell}{\ell}$. Although $\CA$ is defined over $\F_q$ and $\CH$ is defined over $\F_{q^2}$, we shall prove that $\CH$ has a basis over the subfield $\F_q$ and in fact, may be considered as a $\F_q$ code. We begin by defining the following map.

\begin{definition}

Let $\F_{q^m}$ be a finite field containing $\F_q$. If $c \in \F_{q^m}^n$, we define $$c^q := (c_1^q, c_2^q, \ldots, c_n^q).$$

\end{definition}

This definition is extended to codes as follows.

\begin{definition}

Let $\F_{q^m}$ be a finite field containing $\F_q$. If $C  \leq \F_{q^m}^n$, we define $$C^q := \{c^q \ | \ c \in C \}.$$

\end{definition}

We state without proof that $C^q$ is also a $\F_{q^m}$ linear code whenever $C$ is. 

\begin{proposition}\label{prop: Sticht 1}\cite{Sticht}%Stichtenoth 90
Let $\F_{q^m}$ be a finite field containing $\F_q$. Let $C$ be a linear code over $\F_{q^m}$, Then $C$ has a basis over $\F_q$ if and only if $C = C^q$.

\end{proposition}

\begin{proposition}\label{prop: Sticht 2}\cite{Sticht}%Stichtenoth 90
Let $\F_{q^m}$ be a finite field containing $\F_q$. Let $C$ be a linear code over $\F_{q^m}$, If $C = C^q$, then all of its minimum distance codewords are multiples of a codeword in the subfield $\F_q$.

\end{proposition}

\begin{definition}

Let $f \in \Fl$. We denote the conjugate of $f$ by $$ f_{conj} = \sum\limits_{I,J \subseteq [\ell], \# I = \# J} f_{I,J}^q det_{J,I}(X).$$
\end{definition}

\begin{lemma}\label{lem: Fq}
The code $\CH$ satisfies $$\CH = \CH^q.$$
\end{lemma}
\begin{proof}

We only need to prove that if $c \in \CH$ then so is $c^q \in \CH$. Recall that for any Hermitian matrix $H \in \HM$ $det_{J,I}(H) = det_{I,J}(H)^q$. Hence if $f \in \Fl$ satisfies $$ f = \sum\limits_{I,J \subseteq [\ell], \# I = \# J} f_{I,J} det_{I,J}(X).$$

Then for any $P_i \in \HM$, $f^q$ satisfies

$$ f(P_i)^q = \sum\limits_{I,J \subseteq [\ell], \# I = \# J} f_{I,J}^q det_{I,J}(P_i)^q.$$

However, as $P_i$ is Hermitian we have that 

$$ f(P_i)^q = \sum\limits_{I,J \subseteq [\ell], \# I = \# J} f_{I,J}^q det_{J,I}(P_i).$$

Therefore the linear combination $$ f_{conj} = \sum\limits_{I,J \subseteq [\ell], \# I = \# J} f_{I,J}^q det_{J,I}(X)$$ satisfies $$ev(f_{conj}) = ev(f)^q.$$ \end{proof}

Now we calculate the $\F_q$--basis for the code $\CH$.

\begin{lemma}\label{lem:fqbasis}

Let $\{\alpha, \alpha^q\}$ be a basis of $\F_{q^2}$ over $\F_q$. Then $\CH$ is generated as a code over $\F_q$ by the functions $$ev(\alpha  det_{I,J}(X) + \alpha^q det_{J,I}(X)) $$ and $$ev(\alpha^q  det_{I,J}(X) + \alpha det_{J,I}(X))$$ 
\end{lemma}
\begin{proof}

The definition of the code $\CH$ implies that the code $\CH$ is generated by $ev_{\HM}( \alpha f_{I,J} det_{I,J}(X))$ where the coefficient $f_{I,J} \in \F_{q^2}$. 

Note that in the case $I=J$ then $ev_{\HM}(det_{I,J}(X))$ is clearly an $\F_q$--valued function because $X^{I,I}$ is a principal determinant of a Hermitian matrix satisfying $(X^{I,I})^T  = (X^{I,I})^{(q)}$.

If $I\neq J$ then note that the vector space  $$\langle det_{I,J}(X), det_{J,I}(X)\rangle$$ is spanned by $$\langle \alpha det_{I,J}(X) + \alpha^q det_{J,I}(X),  \alpha^q det_{I,J}(X) + \alpha det_{J,I}(X) \rangle .$$ When evaluating $det_{I,J}(X)$ on a nonprincipal minor of a Hermitian matrix we obtain that $det_{I,J}(X)^q = det_{J,I}(X)$. Therefore $ev_{\HM}( \alpha det_{I,J}(X) + \alpha^q det_{J,I}(X))$ and $ev_{\HM}( \alpha^q det_{I,J}(X) + \alpha det_{J,I}(X))$ are $\F_q$ valued functions. Thus we've found a basis for $\CH$ of the requiered form.
 %$ev(\alpha  det_{I,J}(X) + \alpha^q det_{J,I}(X)) $ and ev(\alpha^q  det_{I,J}(X) + \alpha det_{J,I}(X))$ 

\end{proof}

We have the following corollary on the codewords of $\CH$ which take values exclusively on the subfield $\F_q$.

\begin{corollary}\label{col: Fq}
Let $c = ev_{\HM}(f)$ be a codeword of $\CH$. Denote by $c_H$ the position of $c$ indexed by $H \in \HM$. Then $c_H^q = c_H$ $ \forall H \in \HM$ if and only if $f_{I,J}^q = f_{J,I}$. \end{corollary} \begin{proof} Note that for $c = ev_{\HM}(f)$ the position $c_H$ is equal to $f(H)$. Lemma \ref{lem:fqbasis} implies $f(H) \in \F_q$ if and only if $f_{I,J}= f_{J, I}^q.$ \end{proof}

By Proposition \ref{prop: Sticht 2}, we may assume our minimum weight codewords are over $\F_q$. Thus, by Corollary \ref{col: Fq}, we may assume our coefficients meet $f_{I,J}= f_{J, I}^q$.

\begin{definition}
Let $C$ be a code of length $n$. We say that a permutation $\sigma \in S_n$ is an automorphism of $C$ if and only if $$(c_1, c_2, \ldots, c_n) \in C \makebox{ if and only if } (c_{\sigma(1)}, c_{\sigma(2)}, \ldots, c_{\sigma(n)}) \in C. $$

The group of such automorphisms is determined by $Aut(C).$
\end{definition}

We state the following automorphisms of $\CA$ from \cite{BGT2}.

\begin{proposition}\cite[Lemma 7]{BGT2}

The automorphism group $Aut(\CA)$ contains the following permutations:
\begin{itemize}
    \item For $A \in GL_{\ell}(\F_q), \X \mapsto A\X$.
    \item For $B \in GL_{\ell}(\F_q), \X \mapsto \X B$.
    \item For $M \in \mathbb{M}^{\ell \times \ell}(\F_q), \X \mapsto \X + M$.
    \item $\X \mapsto \X^T$.
\end{itemize}

\end{proposition}

Some of the automorphisms of $\CA$ are also automorphisms of $\CH$ as seen in the following lemma.

\begin{lemma}

The automorphism group $Aut(\CH)$ contains the group generated by the following permutations
\begin{itemize}
    \item For $A \in GL_{\ell}(\F_{q^2}), \X \mapsto A^{(q)T}\X A$.
    \item For $M \in \HM, \X \mapsto \X + M$.
    \item $\X \mapsto \X^T$.
\end{itemize}

\end{lemma}

\begin{proof}
Note that $\CH$ is obtained from $\CA$ by removing all matrices which are not Hermitian. That is the code $\CH$ is a puncturing of the code $\CA$ at the matrices in $\mathbb{M}^{\ell \times \ell}(\F_{q^2}) \setminus \HM$ to obtain $\CH$. Therefore the permutations in $Aut(\CA)$ fixing $\HM$ setwise are permutations of $Aut(\CH)$. If $H$ is a Hermitian matrix, then so are the matrices $A^{(q)T}HA$, $H + H'$ where $H' \in \HM$ and $H^T$. Therefore these permutations are automorphisms of $\CH$.
%The aforementioned maps take elements from $\Delta(\ell)$ to $\Delta(\ell)$ and preserve Hermitian matrices. %are automorphisms of the affine Grassmann code
\end{proof}

%\begin{definition}
%    A minor $\mathcal{M}_{A,B}$ with respective row and column sets $A,B$ is maximal if:
%    $\forall \mathcal{M}_{I,J}\in supp(f)$,$ A\nsubseteq I$ or $B\nsubseteq J$.
%\end{definition}
%\\

One of the most important parameters of a linear code is its minimum distance. We state its definition as follows:

\begin{definition}
The \emph{Hamming distance} of the vectors $x = (x_1, x_2, \ldots, x_n)$ and $y = (y_1, y_2, \ldots, y_n)$ is the number of positions in which $x$ and $y$ differ. That is:

$$d(x,y) := \#\{ i \ | \ x_i \neq y_i  \} $$

\end{definition}
For example: let $x = (1010)$ and $y = (1100)$\\
Then $d(x,y) = 2$ because they differ in the second and third position.

\begin{definition}
The \emph{weight} of $x = (x_1, x_2, \ldots, x_n)$ is the number of positions in which $x_i \neq 0$. That is:

$$w(x) := \#\{ i \ | \ x_i \neq 0  \} $$

\end{definition}

Note that the weight of a vector is the same as its distance to the zero vector. That is $w(x) = d(x,0)$. We now state the definition of the minimum distance of a code.

\begin{definition}

Let $C$ be a code. Then the \emph{minimum distance} of $C$ is the minimum number of positions in which any two distinct elements of $C$ differ.

$$d(C) = \min\limits_{x,y \in C} d(x,y). $$

\end{definition}

For linear codes, as is our case, the distance can be calculated in terms of the weights of the vectors.  One of our remarkable findings is that not only that $\CH$ may be considered as a $\F_q$ code, but also that $\CH$ has a much better minimum distance than $\CA$. The distance of affine Grassmann codes (see \cite{BGT}) is  $$d(\CA) = \# GL_\ell(\F_q) = \prod\limits_{i=0}^{\ell-1}q^\ell - q^i.$$ Now we compare with our main result.

\begin{claim}

Suppose that $\ell \geq 2$. Then $$d(\CH)  = q^{\ell^2}- q^{\ell^2-1}- q^{\ell^2-3}.$$ 

\end{claim}

In the next two sections of the paper we work out a proof of this claim by induction. We use polynomial evaluation and bounds from the fundamental theorem of Algebra to determine $d(\CH).$

\section{Determining $d(\CH)$ for $2 \leq \ell \leq 3$}

Our calculation of $d(\CH)$ depends on mathematical induction. The base cases $\ell = 2$ and $\ell = 3$ are the two cases which best illustrate our proof. We include several examples. Once we establish certain bounds on $wt(ev(f))$ over $\HM$, we then derive respective bounds on $wt(ev(f))$ for the Hermitian matrices of size $\ell +1 \times \ell +1$. We begin with the following lemma counting the number of zeroes of a certain quadratic equation.

\begin{lemma}\label{lem:hyperboliczeroes}
 Let $a,b,\lambda \in \F_q$, where $\lambda\neq 0$. Then
 \begin{itemize}
     \item The equation $(X_1+a)(X_2+b) = 0$ has $2q-1$ solutions over $\F_q$.
     \item The equation $(X_1+a)(X_2+b) = \lambda $ has $q-1$ solutions over $\F_q$.
 \end{itemize}
\end{lemma}
\begin{proof}
We begin with the case $(X_1+a)(X_2+b) = 0$. In this case then either $(X_1+a)= 0$ or $(X_2+b) = 0$. If $X_1 = -a$, any of the $q$ values for $X_2$ is a solution to the equation. Similarly, for $X_2 = -b$, any of the $q$ values for $X_1$ is a solution to the equation. The solution $X_1 = -a, X_2 = -b$ is counted twice this implies we have $q+q-1 = 2q-1$ total solutions to the equation.

Now we consider $(X_1+a)(X_2+b) = \lambda \neq 0$. If $X_1 = -a$ then the equation becomes $(a-a)(X_2+b) = \lambda$ which has no solution.  If $X_1 = \alpha$  where $\alpha$ is any element of $\F_q$ except $\alpha = -a$, then for there is exactly one value of   $X_2$ (namely $X_2 = \frac{\lambda}{\alpha+a}-b$) such that the equation is satisfied. %We check that 

%$$ (X_1+a)(X_2+b) = \lambda $$ has $X_1 = \alpha, X_2 = \frac{\lambda}{\alpha+a}-b$ as solutions.

%$$ \left(\alpha+a\right)\left(\frac{\lambda}{\alpha+a}-b+b\right) = \lambda $$
%$$ \left(\alpha+a\right)\left(\frac{\lambda}{\alpha+a}\right) = \lambda $$
%$$ \lambda = \lambda.$$

As for any of $q-1$ values $\alpha \neq -a$ for $X_1$, we find exactly one value of $X_2$ such that $(X_1+a)(X_2+b) = \lambda $ is satisfied, it is established that $(X_1+a)(X_2+b) = \lambda $ has $q-1$ solutions. \end{proof}

We shall also need the following lemma on the number of solutions to a particular system of polynomial equations over $\F_{q^2}$.

\begin{lemma}\label{lem:systemsols}
Let $a_1, a_2, \ldots, a_n \in \F_q$ and for $1 \leq i,j \leq n$ let $b_{i,j} \in \F_{q^2}$. Then the system of polynomial equations given by $$X_i^{q+1} = a_i, 1 \leq i \leq n $$
$$X_iX_j^q = b_{i,j} $$ has at most $q+1$ solutions.
\end{lemma}
\begin{proof}

If the $a_i$'s satisfy $a_i =  0,$ then the only possible solution is $X_i = 0.$ If there is some $a_s \neq 0$, then there are at most $q+1$ values for $X_s$ which satisfy $X_s^{q+1} = a_s$ In this case, for each of the $q+1$ solutions of $X_s^{q+1} = a_s$, there is at most one value (namely $X_s = c_s$, $X_i = \frac{b_{i,s}}{c_s^q}$ which satisfies the equation $X_i X_s^q = b_{i,s}$. Therefore there are at most $q+1$ solutions to the system of equations. \end{proof}

Now we define the support of a combination of minors $f \in \Fl$ and the concept of a maximal ``term". This concept will be akin to the degree in order to make the induction proof for the minimum distance.

\begin{definition}
Let $f \in \Fl$, where $$f = \sum\limits_{I,J \subseteq [\ell], \# I = \# J} f_{I,J} det_{I,J}(X).$$ The support of $f$ is defined as $$supp(f) :=\{ det_{I,J}(X) \ | \ f_{I,J} \neq 0  \}. $$

\end{definition}
Example: Let $\X = \begin{bmatrix}
    X_1&X_2\\X_2^q&X_3
\end{bmatrix}$\\
If $f = 1+X_1X_3-X_2^{q+1}$, then $supp(f) = \{ det_{\emptyset, \emptyset}(X), det_{\{1,2\},\{1,2\}}(X)\}$
\begin{definition}
Let $f \in \Fl$, where $$f = \sum\limits_{I,J \subseteq [\ell], \# I = \# J} f_{I,J} det_{I,J}(X).$$ We say a minor $det_{I,J} \in supp(f)$ is \emph{maximal} if and only if for any other minor $det_{I',J'} \in supp(f)$ we have that $I \not\subseteq I'$ or $J \not\subseteq J'$. That is the columns and rows of $det_{I,J}$ are not contained in the rows and columns of any other determinant in $supp(f)$.
\end{definition}
Example: Let $X = \begin{bmatrix}
    X_1&X_2&X_3\\X_2^q&X_4&X_5\\X_3^q&X_5^q&X_6
\end{bmatrix}$ and $$f = det_{\emptyset, \emptyset}(X) + det_{\{1\}, \{2\} }(X) det_{\{1,2\},\{1,2\}}(X)+ det_{\{1,2\},\{2,3\}}(X)\},$$

and $$supp(f) = \{ det_{\emptyset, \emptyset}(X),det_{\{1\},\{2\}}(X), det_{\{1,2\},\{1,2\}}(X), det_{\{1,2\},\{2,3\}}(X)\}.$$ The minors $det_{\{1,2\},\{1,2\}}(X)$ and $det_{\{1,2\},\{2,3\}}(X)$ are maximal. However, the minors $det_{\emptyset, \emptyset}(X)$ and $det_{\{1\},\{2\}}(X)$ are not maximal, because their row and column sets are contained in the row and column sets of other minors. Moreover, we remark that $f \neq f_{conj}$, because  $det_{\{2,3\},\{1,2\}}(X) \notin supp(f).$

The following lemma relating translations and self--conjugate combinations gives insight on the structure of $f \in \Fl.$ We shall use the automorphisms of the form $f(X) \mapsto f(X +H), H \in \HM$ to simplify the calculation of $wt(ev(f))$. Now we prove that that under a certain translation, we may consider $f$ has no ``terms" of second highest degree.

\begin{lemma}\label{lem:minorclear}

Let $f \in \Fl$ where $f = f_{conj}$. Suppose $det_{I,I}$ is a maximal minor of size $k$ in $f$. There exists a matrix $H \in \HM$ such that $f(\X + H)$ has no $\# I -1 \times \# I -1$ minors in its support whose rows and columns are contained in $I$. %The minors $det_{I \setminus \{i\}, I \setminus \{j\}}(\X +H)$ will also be a linear combination of $det_{I \setminus \{i\}, I \setminus \{j\}}(\X )$ and other minors of smaller size.

\end{lemma}
\begin{proof}
Let $H \in \HM$. If $det_{I,I}(\X + H)$ expanded is via minor expansions along a row $i \in I$ it may be seen that $det_{I,I}(\X + H)$ is a combination of $det_{I,I}(\X)$, $det_{I \setminus \{i\}, I \setminus \{j\}}(\X)$ and other minors of smaller size. However, the minor given by rows $I \setminus \{i\}$ and columns  $I \setminus \{j\}$ appears with coefficient $(-1)^{i+j+1}H_{i,j}$.

Suppose $f = f_{conj}$ and that $det_{I,I}$ is a maximal minor. Note that as $f$ is self--conjugate, for any $i,j \in I$, $f_{I\setminus \{ i \},I \setminus \{j\}} = f_{I\setminus \{ j \},I \setminus \{i\}}^q$. Without loss of generality, we assume $f_{I,I} = 1$. Suppose $H_f \in \HM$ where $h_{i,j} = (-1)^{i+j-1}f_{I \setminus\{i\}, J \setminus \{j\}}.$

Because $f$ is self--conjugate, the matrix is Hermitian. Clearly the expansion of $det_{I,J}(\X+H) + det_{J,I}(\X + H)$ has $det_{I\setminus \{i \}, J \setminus \{j\}}(\X)$ in its support, appearing with coefficient $-f_{I \setminus\{i\}, J \setminus \{j\}}$.

After the expansion of $f(\X +H) $ we see that the $\# I -1 \times \# I -1$ minors from $det_{I,I}(\X + H)$ cancel the $\# I -1 \times \# I -1$ minors in $f$, which implies $f(\X +H)$ has no $\# I -1 \times \#I -1$ minors whose rows and columns are contained in $I$.   \end{proof}

\subsection{Calculating $d(C^{\mathbb{H}}(2))$}

For the case $\ell = 2$, we shall consider a generic matrix $\X \in \HM $ as a matrix of the form %then we shall prove that   $d_2 = (q^{{\ell^2}-3})(q^3-q^2-1) = q^4-q^3-q$\\
$$\X = \begin{bmatrix}
    X_1&X_2\\X_2^q&X_3
    \end{bmatrix},$$ \\
where $X_1, X_3 \in \F_q, X_2 \in \F_{q^2}$    

In this case a function $f \in \Fl$ is of the form
$$f = f_0 + f_{1,1}X_1 + f_{1,2}X_2 + f_{2,1}X_2^q+ f_{2,2}X_3+ f_{\{1,2\},\{1,2\}}(X_1X_3-X_2^{q+1}), f_{i,j} \in \F_{q^2}.$$
More specifically,  as we assume $f = f_{conj}$, Corollary \ref{col: Fq} implies $f$ is of the form $$f = f_0 + f_{1,1}X_1 + f_{1,2}X_2 + f_{1,2}^qX_2^q+ f_{2,2}X_3+ f_{\{1,2\},\{1,2\}}(X_1X_3-X_2^{q+1})$$
Where $f_0, f_{1,1}, f_{2,2},f_{\{1,2\},\{1,2\}} \in \F_q$ and $f_{1,2} \in \F_{q^2}$.

We split our proof in two cases: $f_{\{1,2\},\{1,2\}} = 0$ or $f_{\{1,2\},\{1,2\}} \neq 0$.
\begin{lemma}
Let $\ell = 2$. Suppose $f \in \Fl$ where $f$ is of the form
$$f = f_0 + f_{1,1}X_1 + f_{1,2}X_2 + f_{2,1}X_2^q+ f_{2,2}X_3.$$
Then $wt(ev(f)) \geq q^4-q^3$

\end{lemma}
\begin{proof}

To determine $w(f)$ we count the solutions to $$f = f_0 + f_{1,1}X_1 + f_{1,2}X_2 + f_{2,1}X_2^q+ f_{2,2}X_3 = 0.$$ Suppose that $f_{1,1} \neq 0$. Then, for any of the $q$ values of $X_3$ and any of the $q^2$ values of $X_2$, there are at least $q-1$ values of $X_1$ which make $f \neq 0.$ The case $f_{2,2} \neq 0 $ is similar.

Now we consider $f_{1,1} = f_{2,2} = 0$. As $X_1, X_3$ represent elements of $\F_q$, we may replace them with arbitrary values from $\F_q$. In that case, the equation $f = 0$ becomes $f_{1,2}^qX_2^q + f_{1,2} X_2 = -f_0$. As this is a polynomial of degree at most $q$ over $\F_{q^2}$, it has at most $q$ zeros for each of the $q^2$ total possible values of $X_1$ and $X_3$. Consequently there are at most $q^3$ matrices in $\HM$ such that $f(H) = 0$ and $w(f) \geq q^4-q^3$. 
%This implies that $d_2 \leq q^4-q^3.$
\end{proof}

Now consider $f \in \Fl$ where the $2 \times 2$ minor of $\X$ appears in $supp(f)$. Without loss of generality we assume $f_{\{1,2\},\{1,2\}} = 1$
\begin{lemma}\label{lem:wt22}
Let $\ell = 2$. Suppose $f \in \Fl$ where $f$ is of the form
$$f = f_0 + f_{1,1}X_1 + f_{1,2}X_2 + f_{2,1}X_2^q+ f_{2,2}X_3+ X_1X_3-X_2^{q+1}.$$
Then $wt(ev(f)) \geq q^4-q^3-q$

\end{lemma}
\begin{proof}

As in the previous case we count the number of solutions to 
$$f = f_0 + f_{1,1}X_1 + f_{1,2}X_2 + f_{2,1}X_2^q+ f_{2,2}X_3+ X_1X_3-X_2^{q+1} = 0.$$

We move the terms with $X_1, X_3$ to one side and obtain:

$$f_{1,1}X_1 + f_{2,2}X_3 + X_1X_3 = X_2^{q+1} - f_{1,2}^qX_2^q -f_{1,2}X_2 -f_0.$$

Now we add $f_{1,1}f_{2,2}$ to both sides:

$$f_{1,1}f_{2,2} + f_{1,1}X_1 + f_{2,2}X_3 + X_1X_3 = X_2^{q+1} - f_{1,2}^qX_2^q -f_{1,2}X_2 -f_0 +f_{1,1}f_{2,2}.$$

The left hand side of the equation factors as:

$$(X_1+f_{2,2})(X_3+f_{1,1}) = X_2^{q+1} - f_{1,2}^qX_2^q -f_{1,2}X_2 -f_0 +f_{1,1}f_{2,2}.$$

The right hand side of the equation is an univariate polynomial in $X_2$ of degree $q+1$. Denote by $$P(X_2) =X_2^{q+1} - f_{1,2}^qX_2^q -f_{1,2}X_2 -f_0 +f_{1,1}f_{2,2}.$$ Let $S = \{ \lambda \in \F_{q^2} | P(\lambda) = 0\}$ denote the set of zeroes of $P(X_2)$. Note that $\# S \leq q+1$. Let $\alpha \in S.$
In this case $P(\alpha) = 0$. Lemma \ref{lem:hyperboliczeroes} implies that there are $2q-1$ values of $X_1$ and $X_3$ such that the equation
$$(X_1+f_{2,2})(X_3+f_{1,1}) = P(\alpha)$$ is satisfied. This implies both sides are $0$ for exactly $\#S(2q-1)$ values.

Now assume $\alpha \in \F_{q^2} \setminus S.$ In this case Lemma \ref{lem:hyperboliczeroes} implies that there are $q-1$ values of $X_1$ and $X_3$ such that the equation
$$(X_1+f_{2,2})(X_3+f_{1,1}) = P(\alpha)$$ is satisfied. This implies there are $(q^2-\# S)(q-1)$ solutions to the equation where $P(\alpha) \neq 0.$

Therefore there are $$\# S (2q-1) + (q^2-\#S)(q-1)$$ elements of $\HM$ such that $f = 0.$ As $\# S \leq q+1$, we have that $f$ has at most $$(q+1)(2q-1) + (q^2-q-1)(q-1) = q^3+q$$
Consequently, $w(f) \geq q^4-q^3-q$.
\end{proof}

It is very useful to classify the codewords in $\CA$ for $\ell = 2$ where the $2 \times 2$ determinant is in $supp(f)$.

\begin{lemma}\label{lem:22classfier}
Let $\ell = 2$. Suppose $f \in \Fl$ where $f$ is of the form
$$f = f_0 + f_{1,1}X_1 + f_{1,2}X_2 + f_{1,2}^qX_2^q+ f_{2,2}X_3+ X_1X_3-X_2^{q+1}.$$ The following statements are true:

If $f_0 + f_{1,2}^{q+1}-f_{1,1}f_{2,2} = 0 $ then  $wt(f) = q^4 -q^3 +q^2-q.$

If $f_0 + f_{1,2}^{q+1}-f_{1,1}f_{2,2} \neq 0 $ then  $wt(f) = q^4-q^3-q$.

\end{lemma}

\begin{proof}

Suppose $f$ is of the form $$f = f_0 + f_{1,1}X_1 + f_{1,2}X_2 + f_{1,2}^qX_2^q+ f_{2,2}X_3+ X_1X_3-X_2^{q+1}.$$ As in the proof of Lemma \ref{lem:wt22}, we write $f$ as $$f = (X_1 +f_{2,2})(X_3 +f_{1,1}) + P(X_2) $$ where $$P(X_2) = X_2^{q+1} + f_{1,2}^qX_{2}^q +f_{1,2}X_{2} +f_0 -f_{1,1}f_{2,2} $$ is a polynomial of degree $q+1$.

%ltiplications, we obtain:
%$$f = f_0 + f_{1,1}T_1-f_{1,1}f_{2,2} + f_{1,2}T_2 + f_{1,2}^{q+1} + f_{1,2}^qT_2^q + f_{1,2}^{q+1}+ f_{2,2}T_3 + T_1T_3 - f_{2,2}T_3 -f_{1,1}T_1 - T_2^{q+1} - f_{1,2}T_2 -f_{1,2}^qT_2^q + f_{1,2}^{q+1}.$$

%Cancelling out like terms we obtain:

%$$f = f_0+  f_{1,1}f_{2,2} + f_{1,2}^{q+1} + T_1T_3 - T_2^{q+1}.$$
Now we shall change the variable $X_2$ to the variable $T_2$ where $X_2 = T_2 - f_{1,2}^q$. In this case $$ P(T_2- f_{1,2}^q) = (T_2 - f_{1,2}^q)^{q+1} + f_{1,2}^q(T_2 - f_{1,2}^q)^q +f_{1,2}(T_2 - f_{1,2}^q) +f_0 -f_{1,1}f_{2,2}.$$

Expanding and eliminating like terms we obtain

$$P(T_2- f_{1,2}^q) = T_2^{q+1} - f_{1,2}^{q+1}+f_0 - f_{1,1}f_{2,2}. $$
Note that if $f_{1,2}^{q+1}-f_0 + f_{1,1}f_{2,2} = 0$, then $P(T_2 - f_{1,2}^q)$ has exactly one zero, namely $T_2 = f_{1,2}^q$.

Since $f = f_{conj}$ we assume $f_0, f_{1,1}, f_{2,2}$ are in $\F_q$. As $f_{1,2}^{q+1}-f_0 + f_{1,1}f_{2,2} \neq 0$ it assumes values over $\F_q$. Hence $P(T_2 - f_{1,2}^q)$ has $q+1$ zeroes over $\F_{q^2}$. Note that the number of zeroes of $P(T_2 - f_{1,2}^q)$ is precisely the same number of solutions to $P(X_2) = 0$

When $P(X_2) = 0$, there are $2q-1$ values of $X_1$ and $X_3$ which make $f = 0$. When $P(X_2) \neq 0$, there are $q-1$ values of $X_1$ and $X_3$ which make $f = 0$.

Therefore, $f = 0$ for $2q-1 + (q^2-1)(q-1) = q^3 - q^2 -q +1 +2q -1 =q^3-q^2+q$ and $wt(f) = q^4 -q^3 +q^2-q.$

If $f_{1,2}^{q+1}-f_0 + f_{1,1}f_{2,2} \neq 0$, then there are $q+1$ values of  $X_2$ which make $P(X_2) = 0$. There are $2q-1$ values of $X_1$ and $X_3$ which make $f = 0$. When $P(X_2)  \neq 0$, there are $q-1$ values of $X_1$ and $X_3$ which make the equation true. Therefore $f = 0$ for $(2q-1)(q+1) + (q^2-q-1)(q-1) = q^3 +q $ and $wt(f) = q^4 -q^3 -q.$

\end{proof}

\subsection{Calculating $d(C^{\mathbb{H}}(3))$}

The technique used in \cite{BGT} to find the minimum distance of $\CA$ was to specialize from the $\ell \times \ell$ case down to the $(\ell-1) \times (\ell -1 )$ case. As the matrices in the affine Grassmann code are generic, one can perform a partial evaluation on any row and any column while still preserving the structure of the code $\CA$. In the case of Hermitian matrices, the Hermitian property of the matrices must be preserved. This means that a partial evaluation on a column also fixes the corresponding row. In order to refine the concept of the size of a minor and to simplify our induction proof we introduce the following definition.

\begin{definition}
Let $I, J \subseteq [\ell]$. We define the \emph{spread} of the minor $X^{I,J}$ as the set $I \cup J$. 

\end{definition}

 Consider the matrix $X = \begin{bmatrix}
    X_1&X_2& Y_1\\X_2^q&X_3 & Y_2\\ Y_1^q&Y_2^q&Y_3
    \end{bmatrix}$ and the minor given by rows $\{1,2\}$ and columns $\{2,3\}$. That is the minor $f = det_{\{1,2\},\{2,3\}}(X) =  \begin{vmatrix}
    X_2& Y_1\\X_3 & Y_2
    \end{vmatrix}$. The spread of the minor is  $I \cup J = \{1,2\} \cup \{2,3\} =  \{1,2,3\}$. The following lemma will prove that in several cases we can view the $\ell = 3$ case as several $\ell = 2$ cases.

\begin{lemma}
Let $f \in \Fl$. Suppose that $\ell = 3$ and $f$ has a maximal minor whose spread has size $\leq 2$. Then $wt(f) \geq q^9 -q^8-q^6$.
\end{lemma}
\begin{proof}
Let $f$ be as in the statement of the lemma. There is a row and a column which does not appear in the spread of the maximal minor. Then for any of the $q^5$ values one can put on this column, $f$ specializes to a combination of $2\times2$ determinants with the same maximal minors (though others may be changed due to the specialization). As each specialization has weight at least $q^4-q^3-q$ and there are $q^5$ specializations, the statement follows.\end{proof}

To calculate the distance of our code, we need to count the invertible Hermitian matrices. We shall use group actions with the following sets:
\begin{itemize}
    \item The set of $\ell \times \ell$ Hermitian matrices \\ $${HL}_{\ell}(\mathbb{F}_{q^2}) = \{H \in \mathbb{H}^{\ell \times \ell}(\mathbb{F}_{q^2})| det(H) \neq 0 \}$$ 
    \item The General Linear group over $\F_{q^2}$,  \\ $${GL}_{\ell}(\mathbb{F}_{q^2}) = \{M \in \mathbb{M}^{\ell \times \ell}(\mathbb{F}_{q^2})| det(M) \neq 0 \}$$ 
    \item The group of Unirary matrices over $\F_{q^2}$ \\ $${U}_{\ell}(\mathbb{F}_{q^2}) = \{U\in {GL}_{\ell}(\mathbb{F}_{q^2}) | U\times U^{*} = I \}$$
    %\item ${U}_{\ell}(\mathbb{F}_{q^2}) = \{U \in {GL}_{\ell}(\mathbb{F}_{q^2})|(A\in {HL}_{\ell}(\mathbb{F}_{q^2}))  [U^{*}^{T}AU = A]  \}$
\end{itemize}
The cardinality of the classic finite groups is known:
\begin{itemize}
    \item $$\#{GL}_{\ell}(\mathbb{F}_{q^2}) = q^{2\binom{\ell}{2}}\prod_{i=1}^{\ell} (q^{2i}-1)$$
    \item $$\# {U}_{\ell}(\mathbb{F}_{q^2}) = q^{\binom{\ell}{2}}\prod_{i=1}^{\ell} (q^i-(-1)^i)$$
\end{itemize}
\begin{proposition}\label{prop: invertible hermite}
 The cardinality of ${HL}_{\ell}(\mathbb{F}_{q^2})= \{H \in \mathbb{H}^{\ell \times \ell}(\mathbb{F}_q)| det(H) \neq 0 \}$ is given by:
 $ q^{\binom{\ell}{2}}\prod_{i=1}^{\ell} (q^i+(-1)^i). $
\end{proposition}
\begin{proof}

%Note that by \cite{fulton_1977}, ${HL}_{\ell}(\mathbb{F}_{q^2})$ is the orbit of ${GL}_{\ell}(\mathbb{F}_{q^2})$ when ${U}_{\ell}(\mathbb{F}_{q^2})$ acts under the group action defined as $U \cdot A = U*AU$

Recall that ${GL}_{\ell}(\mathbb{F}_{q^2})$ acts upon ${HL}_{\ell}(\mathbb{F}_{q^2})$ under the group action defined as $G \cdot A = G^*AG$. By \cite{fulton_1977}, this action has exactly one orbit. Moreover, the orbit stabilizer, that is the subgroup that fixes the identity, is ${U}_{\ell}(\mathbb{F}_{q^2})$. 
%\\
%Let our group be ${GL}_{\ell}(\mathbb{F}_{q^2})$, our set be ${HL}_{\ell}(\mathbb{F}_{q^2})$ and our group action $\alpha$ be $\alpha(X,Y) = X^{*}YX$. By \cite{fulton_1977}, for any hermitian matrix A, there exists an invertible matrix $P \in {GL}_{\ell}(\mathbb{F}_{q^2})$ such that $PAP^{*} = I_k$. This implies that for $H\in {HL}_{\ell}(\mathbb{F}_{q^2})$, there exists an invertible matrix $P \in {GL}_{\ell}(\mathbb{F}_{q^2})$ such that $PAP^{*} = I_\ell$. This implies that the action $\alpha$ is a transitive group action on ${HL}_{\ell}(\mathbb{F}_{q^2})$. Consequently, by taking $A = I$, we obtain that ${U}_{\ell}(\mathbb{F}_{q^2})$ is the stabilizer group when using $I_\ell$. %That is that ${U}_{\ell}(\mathbb{F}_{q^2})$ is the subgroup of G which fixes X. 
%Because the action is transitive, there is exactly one orbit, that is ${HL}_{\ell}(\mathbb{F}_{q^2}) $.\\ 
By the Orbit-Stabilizer Theorem:
%\begin{equation}
%    \#G \cdot X = \frac{\#G}{\#G_x}
%\end{equation}
\begin{equation}
   \#{HL}_{\ell}(\mathbb{F}_{q^2}) =\frac{\#{GL}_{\ell}(\mathbb{F}_{q^2})}{\#{U}_{\ell}(\mathbb{F}_{q^2})}
\end{equation}
\begin{equation}
    =\frac{q^{2\binom{\ell}{2}}\prod_{i=1}^{\ell} (q^{2i}-1)}{q^{\binom{\ell}{2}}\prod_{i=1}^{\ell} (q^i-(-1)^i)}
\end{equation}
\begin{equation}
    =\frac{q^{2\binom{\ell}{2}}\prod_{i=1}^{\ell} (q^{i}+1)(q^{i}-1)}{q^{\binom{\ell}{2}}\prod_{i=1}^{\ell} (q^i-(-1)^i)}
\end{equation}
\begin{equation}
    =q^{\binom{\ell}{2}}\prod_{i=1}^{\ell} (q^i+(-1)^i)
\end{equation}
\end{proof}
Now we shall assume $\X$ is of the form

   $$\X = \begin{bmatrix}
    X_1&X_2& Y_1\\X_2^q&X_3 & Y_2\\ Y_1^q&Y_2^q&Y_3
    \end{bmatrix},$$ 
    where $X_1, X_3, Y_3$ are variables of elements in $\F_q$ and $X_2, Y_1, Y_2$ are variables of elements in $\F_{q^2}$. Now we shall study the $q^5$ possible specializations of $Y_1, Y_2, Y_3$ and their effect on $f \in \Fl$ on the remaining unspecialized $2 \times 2$ submatrix. We remind the reader of the following minor expansion:
    
    \begin{proposition}
$$    det(\X) = det_{\{1,2\}, \{1,2\}}(\X)Y_3 -det_{\{1,3\}, \{1,2\}}(\X)Y_2 + det_{\{2,3\}, \{1,2\}}(\X)Y_1   $$
    \end{proposition}
    
%As the set $\HM$ is a set closed under addition, a matrix $H \in \HM$  will permute the set $\HM$ via the map $M \mapsto M + H.$ The matrix $H \in \HM$ also induces a permutation of $\Fl$ via the map $f(\X) \mapsto f(\X + H).$ We shall use Hermitian translations to transform the codewords and simplify the calculation of $wt(f)$. 

Let us denote by $f_{a,b,c}(\X)$ the minor combination obtained by the partial evaluation of $f(\X)$ at $Y_1 = a$, $Y_2= b$ and $Y_3 = c$.

\begin{lemma}\label{lem:fulleval}
Let $f \in \Fl$. Suppose that for all $q^5$ partial evaluations we have that $f_{a,b,c}(\X) \neq 0.$ Then $wt(f) \geq q^9-q^8-q^6$

\end{lemma}
\begin{proof}

Suppose that $f$ is as in the hypothesis of the lemma. Then for any given partial evaluation, $f_{a,b,c}(\X)$ is a  nonzero combination of minors in the $2 \times 2$ case. As $wt(f_{a,b,c}) \geq q^4-q^3-q$ and all $q^5$ combinations are nonzero, the result follows.
\end{proof}
%In the case one of the evaluations is the zero function we can prove the minors appearing in the support are special.

%\begin{lemma}
%Let $f \in \Fl$. If there exists a partial evaluation along the last row such that $f_{a,b,c}(\X) = 0$ then there is a translation $g = f(\X + H)$ such that $g_{0,0,0}(\X) = 0.$
%\end{lemma}
%\begin{proof}
%Take $H$ to be any Hermitian matrix with $-a,-b,-c$ on the last column. \end{proof}
%A function $f \in \Fl$ which is zero whenever the entries on the third column are zero will have only certain minors occurring on its support.
%\begin{lemma}
%Let $f \in \Fl$. Suppose that $f_{0,0,0}(\X) = 0$ then all minors in $supp(f)$ contain either row $3$ or column $3$
%\end{lemma}
%\begin{proof}
%If $f$ has any minor not containing row $3$ or column $3$, it would not be affected by the partial evaluation, which then implies $f_{0,0,0} \neq 0.$ \end{proof}

We are now ready to find $wt(f)$ for $\ell = 3.$ We begin with the full $3 \times 3$ determinant first.  %We may assume without loss of generality that $f_{0,0,0} = 0$ and therefore all minors of $supp(f)$ have full spread $\{1,2,3\}$.

\begin{lemma}\label{lem:minorclear3}
Let $f \in \Fl$ where $f = f_{conj}$. Suppose that the maximal minor of $f$ is the full determinant $det_{\{1,2,3\}, \{1,2,3\} }(\X)$. Then we may assume $f$ is of the form $det_{\{1,2,3\}, \{1,2,3\} }(\X) + a_1 X_{1,1} + a_2 X_{2,2} + a_3 X_{3,3} + a_4$.
\end{lemma}

\begin{proof}

Recall $f_{\{1,2,3\}, \{1,2,3\} }\neq 0$ implies we may assume without loss of generality that $f_{\{1,2,3\}, \{1,2,3\} } = 1$. Then, Lemma \ref{lem:minorclear} implies that we may assume without loss of generality that $f$ has no $2 \times 2$ minors.  As $f  = f_{conj}$ we have that $f_{i,j} = f_{i,j}^q$. Moreover, we may express $f_{i,j}X_a+f_{i,j}^qX_a^q$ as $Tr(f_{i,j}X_a)$. This implies that $f$ is of the form: 
%$$\begin{bmatrix}
 %   X_1&X_2& Y_1\\X_2^q&X_3 & Y_2\\ Y_1^q&Y_2^q&Y_3
  %  \end{bmatrix} +    f_{1,1} X_1 + f_{1,2}X_2 +f_{1,2}^qX_2^q + f_{2,2}X_3 +f_{1,3} Y_1 + f_{1,3}^qY_1^q + f_{2,3}Y_2 +f_{2,3}^qY_2^q+ f_{3,3}Y_3 + \delta$$ where $f_{1,1}, f_{2,2}, f_{3,3} \in \F_q$, $f_{1,2}, f_{1,3}, f_{2,3} \in \F_{q^2}$.
    
    $$\begin{bmatrix}
    X_1&X_2& Y_1\\X_2^q&X_3 & Y_2\\ Y_1^q&Y_2^q&Y_3
    \end{bmatrix} +    f_{1,1} X_1 + Tr(f_{1,2}X_2) + f_{2,2}X_3 +Tr(f_{1,3} Y_1) + Tr(f_{2,3}Y_2) + f_{3,3}Y_3 + f_0$$ where $f_{1,1}, f_{2,2}, f_{3,3}, f_0 \in \F_q$, $f_{1,2}, f_{1,3}, f_{2,3} \in \F_{q^2}$.
    
    If we let $F = \begin{bmatrix}
    f_{1,1}&f_{1,2}& f_{1,3}\\f_{1,2}^q& f_{2,2} & f_{2,3}\\ f_{1,3}^q&f_{2,3}^q& f_{3,3}
    \end{bmatrix} $ then $F$ is a Hermitian matrix and $f$ is of the form  $$f(\X) = det_{\{1,2,3\},\{1,2,3\}}(\X) + Tr(F^T \X) +f_{0}.$$
    
    If $F$ is a Hermitian matrix of rank $r$ there exists a nonsingular matrix $A$ such that $A^{(q)T} I_r A = F$.
    
    We rewrite 
    $$f(\X) = det_{\{1,2,3\},\{1,2,3\}}(\X) + Tr(A^{(q)t}I_rA \X) +f_{0}.$$
    
    Recall that any nonsingular matrix $A^{-1}$ induces a permutation of $\HM$ via the map $H \mapsto A^{-1}H (A^{-1})^{(q)t}$. We change the matrix of variables $\X$ to $\Y$ where $\X = (A^{-1})\Y (A^{-1})^{(q)t}$. 
    
    In this case $f$ becomes 
    
        $$f(\Y) =  \frac{1}{a^{q+1}}det_{\{1,2,3\},\{1,2,3\}}(\Y) + Tr(A^{(q)t}I_rA (A^{-1})\Y (A^{-1})^{(q)t} ) +f_{0}.$$
        
        Which in turn it equals
        $$f(\Y) =  \frac{1}{a^{q+1}}det_{\{1,2,3\},\{1,2,3\}}(\Y) + Tr(A^{(q)t}I_r\Y (A^{-1})^{(q)t} ) +f_{0}.$$

        As the matrix $A^{(q)t}I_r\Y (A^{-1})^{(q)t} $ is similar to $I_r \Y$, we have that 
                $$f(\Y) =  \frac{1}{a^{q+1}}det_{\{1,2,3\},\{1,2,3\}}(\Y) + Tr(I_r\Y ) +f_{0}$$ and the result follows.   
    
\end{proof}

\begin{lemma}\label{lem:33classifier}
Let $f \in \Fl$. Suppose $f = f_{conj}$. Suppose that the maximal minor of $f$ is the full determinant $det_{\{1,2,3\}, \{1,2,3\} }(\X)$. Then $$ wt(f) \geq q^9-q^8-q^6+q^5-q^4+q^3  $$
\end{lemma}
\begin{proof}

Lemma \ref{lem:minorclear3} implies $f$ is of the form

$$det_{\{1,2,3\}, \{1,2,3\} }(\X) + f_{1,1} X_1 + f_{2,2} X_3 + f_{3,3} Y_3 + f_0.$$

We shall consider what happens when we evaluate $f$ along the third row and column.

Suppose now that $f_{a,b,c}$ is the partial evaluation of $f$ with $Y_1 = a$, $Y_2 = b$ and $Y_3 = c \neq 0.$ Then $f_{a,b,c}$ looks as such:

$$c(X_1X_3 - X_2^{q+1}) + Tr(a^qbX_2) -b^{q+1}X_1 -a^{q+1}X_3 +  f_{1,1} X_1 + f_{2,2} X_3 + f_{3,3}c + f_0.$$

Now we shall apply Lemma \ref{lem:22classfier} to determine the weight of each partial evaluation. Lemma \ref{lem:22classfier} implies that if the coefficients of $f_{a,b,c}$ satisfy $$\left(\frac{a^qb}{c}\right)^{q+1} - \frac{f_{1,1}-b^{q+1}}{c} \frac{f_{2,2}-a^{q+1}}{c} + f_{3,3} + \frac{f_0}{c} = 0,  $$
where
\begin{itemize}
    \item  Coefficient of $X_1$: $\frac{f_{1,1}-b^{q+1}}{c}$
    \item Coefficient of $X_2$: 
    $\left(\frac{a^qb}{c}\right)$
    \item Coefficient of $X_2^q$: 
    $\left(\frac{a^qb}{c}\right)^{q}$
    \item Coefficient of $X_3$: $   \frac{f_{2,2}-a^{q+1}}{c}$
    \item Constant term:
    $f_{3,3} + \frac{f_0}{c}$
    
\end{itemize}then the partial evaluation has weight $q(q-1)(q^2+1)$, and otherwise it has weight $q^4-q^3-q$.

After evaluating the parenthesis, and multiplying by $c^2$ we obtain

$$a^{q+1}b^{q+1} - (f_{1,1}-b^{q+1})(f_{2,2} - a^{q+1}) + c^2f_{3,3} + cf_0 = 0 $$
or
$$-f_{1,1}f_{2,2}+f_{1,1}a^{q+1}+f_{2,2}b^{q+1}+ c^2f_{3,3} + cf_0 = 0.$$

Now we shall consider the different cases, depending on if $f_{1,1} = 0$ or $f_{1,1} \neq 0$.

If $f_{1,1} = 0$, then Lemma \ref{lem:minorclear3} implies we may assume $f_{2,2} = f_{3,3} = 0$ too. In this case $f = det_{\{1,2,3\}, \{1,2,3\}}(\X) + f_0 = 0$.

If $f_0 = 0$, by Proposition \ref{prop: invertible hermite} we have $$wt(f) = q^3(q-1)(q^2+1)(q^3-1) = q^9 -q^8 + q^7 -2q^6 -q^4 +q^3 .$$
If $f_0 \neq 0$, we count the $3 \times 3$ Hermitian matrices such that $det_{\{1,2,3\}, \{1,2,3\} }(H)  = - f_0 \neq 0$ for $f_0 \in \F_q.$ Recall that, by hypothesis, $f$ is self--conjugate which implies $f_0 = f_0^q$. Recall that the number of $3 \times 3$ Hermitian matrices with $det_{\{1,2,3\}, \{1,2,3\}}(H) \neq 0$ is $$(q-1)(q^2+1)(q^3-1)q^3.$$ The number of such matrices where $det_{\{1,2,3\}, \{1,2,3\}}(H) = -f_0$ is $(q^2+1)(q^3-1)q^3$ which implies   then $wt(f) = q^9 - q^3(q^2+1)(q^3-1) = q^9 -q^8 -q^6+q^5 +q^3$.

There is always one value of $a$ and $b$ which make 

If $f_{1,1} \neq 0$, but $f_{2,2} = f_{3,3} = 0$, if $f_0 = 0$ there is one value of $a$ such that $$f_{1,1}f_{2,2}+f_{1,1}a^{q+1}+f_{2,2}b^{q+1}+ c^2f_{3,3} + cf_0 = 0.$$
In this case there are $(q-1)q^2(q^2-1)$ partial evaluations of weight $q^4-q^3-q$ and $(q-1)q^2$ partial evaluations of weight $q(q-1)(q^2+1)$.

If $f_0 \neq 0$ there are $q+1$ values of $a$ such that $$f_{1,1}f_{2,2}+f_{1,1}a^{q+1}+f_{2,2}b^{q+1}+ c^2f_{3,3} + cf_0 = 0.$$
In this case there are $(q-1)q^2(q^2-q-1)$ partial evaluations of weight $q^4-q^3-q$ and $(q-1)q^2(q+1)$ partial evaluations of weight $q(q-1)(q^2+1)$.

%For each value of $a$, there are at most $q+1$ values of $b$ which make the equation true. Therefore for each of the $q-1$ values of $c$ there are at leas

%If $f_{1,1}, f_{2,2} \neq 0$,  $$f_{1,1}f_{2,2}+f_{1,1}a^{q+1}+f_{2,2}b^{q+1}+ c^2f_{3,3} + cf_0 = 0.$$

%implies that we may assume without loss of generality that $f$ has no $2 \times 2$ minors. This implies that $f_{a,b,c}$ is of the form: 
%$$\begin{bmatrix}
 %   X_1&X_2& Y_1\\X_2^q&X_3 & Y_2\\ Y_1^q&Y_2^q&Y_3
  %  \end{bmatrix} +    f_{1,1} X_1 + f_{1,2}X_2 +f_{1,2}^qX_2^2 + f_{2,2}X_3 +f_{1,3} Y_1 + f_{1,3}^qY_1^q + f_{2,3}Y_2 +f_{2,3}^qY_2^q+ f_{3,3}Y_3 + \delta$$ where $f_{1,1}, f_{2,2}, f_{3,3} \in \F_q$, $f_{1,2}, f_{1,3}, f_{2,3} \in \F_{q^2}$.
    
%Any given specialization $Y_1 = a$, $Y_2 = b$, $Y_3 = c \neq 0$ has the $2 \times 2$ determinant as its maximal term. In this case there are $q^5-q^4$ specializations with each having at least $q^4-q^3-q$ nonzero values. 

Now we shall count the number of zeroes of $f$ when specializing $X_{3,3} = 0.$ The specialization $X_{1,3} = a, X_{2,3} = b, X_{3,3} = 0$ is of the form 

$$f_{a,b,0}(\X) = \begin{bmatrix}
    X_{1}&X_{2}& a\\X_{2}^q&X_{3} & b\\ a^q&b^q&0
    \end{bmatrix} +    f_{1,1} X_{1} + f_{2,2} X_{3} + f_{3,3} (0) + f_0.$$

Expanding the $3 \times 3$ determinant we obtain:

$$f_{a,b,0}(\X) = a^qbX_{2} + ab^qX_{2}^q - a^{q+1}X_3 - b^{q+1}X_1 +    f_{1,1} X_1 + f_{2,2}X_{3} + f_0.$$ 

%Rearranging the coefficients we obtain:

%$$f_{a,b,0} = X_1(f_1 - a^{q+1}) + X_3(f_3 - b^{q+1}) + (ab^q+f_2^q)X_2^q + (a^qb+f_2)X_2 +\delta_{a,b,0}.$$

%Let us suppose $a \neq 0$ (the case $b \neq 0$ is done in a similar way. If we consider the translation given by $X_3 \mapsto X_3 + \frac{f_3}{a^{q+1}}$, then the expression $$ = a^qbX_2 + ab^qX_2^q - a^{q+1}X_3 - b^{q+1}X_1 +    f_1 X_1 + f_2X_2 +f_2^qX_2^2 + \frac{f_3^2}{a^{q+1}} + \delta_{a,b,0}$$

%has $X_3$ as a maximal $1 \times 1$ minor. Therefore it is nonzero for $q^4-q^3$ values of $X_1, X_2, X_3$
Note that the coefficient of the $(1,1)$--minor of the partial specialization $f_{a,b,0}$ is $f_{1,1} - b^{q+1}$,  the coefficient of the $(2,2)$--minor is $f_{2,2} - a^{q+1}$, and the coefficients of the $(1,2)$--minor and the $(2,1)$--minor are $a^qb$ and $ab^q $ respectively. Lemma \ref{lem:systemsols} implies there are at most $q+1$ partial specializations such that all three coefficients are $0$. Therefore, for $q^4-q-1$ specializations we get a nonzero polynomial with at most $q^3$ zeroes.

In the previous cases we obtain

$$wt(f) \geq (q-1)q^2(q^2-1)(q^4-q^3-q)+(q-1)q^2(q(q-1)(q^2+1)) + (q^4-q-1)(q^4-q^3) $$
$$wt(f) \geq q^9-q^8-q^6+q^5-q^4+q^3 $$

\end{proof}

\begin{lemma}
Let $f \in \Fl$. Suppose that $f$ has no maximal $3 \times 3$ minors in its support. Nor that it has a $2 \times 2$ principal minor in its support. It does have $2 \times 2$ nonprincipal minors. Moreover, $f = f_{conj}$. Then $wt(f) \geq q^9-q^8-q^6+q^5$. 

\end{lemma}
\begin{proof}
    We may assume that $f_{0,0,0} = 0$ and that all determinants in $supp(f)$ contain either row $3$ or column $3$.  For any partial specialization $f_{a,b,c}$ we may assume without loss of generality that $f_{a,b,c}(\X)$ is a polynomial with terms $X_1, X_2, X_2^q$ and $X_3$ (and a constant over $\F_q$). Without loss of generality, we assume $f_{\{1,2 \}, \{2,3 \}} \neq 0$

 %$$ f_{a,b,c}(\X) = \alpha(bX_1-a{X_2}^q)+\alpha^q(b^qX_1-a^qX_2)+\beta(bX_2-aX_3)+\beta^q(b^q{X_2}^q-a^qX_3)+\gamma(X_2c-ab^q)+\gamma^q({X_2}^qc-a^qb) + \delta$$
 
 %where $\alpha = f_{\{1,2 \}, \{2,3 \}}, \beta = f_{\{1,2 \}, \{2,3 \}}, \gamma = f_{\{1,3 \}, \{2,3 \}}, \delta \in \F_q $. As $f$ has a leading $2 \times 2$ principal minor at least one of $\alpha, \beta, \gamma \neq 0$. Without loss of generality, we assume $\alpha \neq 0$. This implies f = $\alpha(bX_1-a{X_2}^q)+\alpha^q(b^qX_1-a^qX_2)+g(X)$\\
    %This implies f = $X_1Tr(\alpha b)-Tr(\alpha a X_2)+g(X)$.
%Note that the coefficients of the variables $X_1, X_2, X_2^q,X_3$ in the polynomial $f_{a,b,c}(\X)$ are as follows:

\begin{itemize}
    \item  Coefficient of $X_1$: $f_{\{1,2 \}, \{2,3 \}} b +f_{\{1,2 \}, \{2,3 \}}^qb^q$
    \item Coefficient of $X_2$: $f_{\{1,2 \}, \{2,3 \}}^q a^q + f_{\{1,2 \}, \{2,3 \}} b +f_{\{1,3 \}, \{2,3 \}} c  $
    \item Coefficient of $X_2^q$: $f_{\{1,2 \}, \{2,3 \}} a + f_{\{1,2 \}, \{2,3 \}}^q b^q +f_{\{1,3 \}, \{2,3 \}}^q c^q  $
    \item Coefficient of $X_3$: $   f_{\{1,2 \}, \{2,3 \}} a -f_{\{1,2 \}, \{2,3 \}}^q a^q$
    
\end{itemize}

We shall compute $wt(f)$ by splitting the specializations into two exclusive cases:  $Tr(f_{\{1,2 \}, \{2,3 \}} b) \neq 0$ and $Tr(f_{\{1,2 \}, \{2,3 \}} b) = 0 $   

\textbf{Case 1: $Tr(f_{\{1,2 \}, \{2,3 \}} b) \neq 0$}: 

If $Tr(f_{\{1,2 \}, \{2,3 \}} b) \neq 0 $, then for any of the $q^6$ specializations possible for $X_2, X_3, a, c$, we can find $q-1$ values of $X_1$ such that $f \neq 0.$ Note we have $q^2-q$ values of $b$ such that $Tr(f_{\{1,2 \}, \{2,3 \}} b)\neq 0$ and we have $q^6$ remaining values for $X_2,X_3,a,c$ and $q-1$ values for $X_1$.
This implies there are at least $q^6(q^2-q)(q-1)$ nonzero values in case 1.

\textbf{Case 2: $Tr(f_{\{1,2 \}, \{2,3 \}} b) = 0$}:

If $Tr(f_{\{1,2 \}, \{2,3 \}} b) = 0$, there are $q^2-1$ values for $a$, such that $a \neq 0$. Thus the polynomial $Tr(f_{\{1,2 \}, \{2,3 \}} a X_2) \neq 0$ for $q^2-q$ values of $X_2$. This holds for any of the $q^3$ possible specialization of $X_3, X_1, c$. For the $q^2-1$ nonzero values of $a$, the partial specialization $f_{a,b,c}$ is a polynomial of at most degree $q$ in $X_2$. Thus, $f_{a,b,c} \neq 0$ for at least $q^2- q$ values of $X_2$. Therefore together with the $q$ values of $b$ where $Tr(f_{\{1,2 \}, \{2,3 \}} b) = 0$, thus there are at least $q^4(q^2-1)(q^2-q)$ nonzero values in case 2.

Adding the two exclusive cases together where either $Tr(f_{\{1,2 \}, \{2,3 \}} b) = 0$ or $Tr(f_{\{1,2 \}, \{2,3 \}} b) \neq 0 $ we get obtain: $$w(f) \geq q^6(q^2-q)(q-1)+q^4(q^2-1)(q^2-q).$$
%$w(f) \geq q(q^4-q)(q^4-q^3)$\\
%$w(f) \geq (q^5-q^2)(q^4-q^3)$\\
$$w(f) \geq q^9-q^8-q^6+q^5.$$ 

%The case $f_{\{1,2 \}, \{2,3 \}} \neq 0$ is done in the same way as $f_{\{1,2 \}, \{2,3 \}}  = 0$. If  $f_{\{1,2 \}, \{2,3 \}} = f_{\{1,2 \}, \{2,3 \}} = 0$, then 
% $$ f_{a,b,c}(\X) = f_{\{1,3 \}, \{2,3 \}}(X_2c-ab^q)+f_{\{1,3 \}, \{2,3 \}}^q({X_2}^qc-a^qb) + \delta c$$ where $f_{\{1,3 \}, \{2,3 \}} \neq 0.$ When $f_{\{1,2 \}, \{2,3 \}} = f_{\{1,2 \}, \{2,3 \}} = 0$, we split the calculation on two cases: $c \neq 0$ or $c = 0$.

%\textbf{Case 3: $c \neq 0$} 

%If $c \neq 0$, there are $q^2-q$ values of $X_2$ such that $f_{a,b,c}(\X) \neq 0$. Therefore there are at least  $(q-1)(q^2-q)q^6$ nonzero values of $f$.

%\textbf{Case 4: $c = 0$} 

%If $c = 0$, then $f_{a,b,0} \neq 0$ if  $Tr(f_{\{1,3 \}, \{2,3 \}} ab^q) \neq 0$. In this case  $a$ can take any nonzero value and $b$ has to take one of $q^2-q$ nonzero values.
 
% Therefore $f$ has at least  $(q^2-1)(q^2-q)q^4$ nonzero values. Adding the two exclusive cases together as before we get $$w(f) \geq q^6(q^2-q)(q-1)+q^4(q^2-1)(q^2-q).$$
%$w(f) \geq q(q^4-q)(q^4-q^3)$\\
%$w(f) \geq (q^5-q^2)(q^4-q^3)$\\
%$$w(f) \geq q^9-q^8-q^6+q^5.$$    
\end{proof}

\section{Finding $d(\CH)$ for $\ell \geq 4$}
Now we generalize some of the previous arguments for any $\ell$. These generalizations will establish certain relations between $\CH$ and $C^\mathbb{H}(\ell+1)$. These relations determine the minimum distance in the general case. First, we recall the following notation:
\begin{definition}
We denote the elementary matrix for row addition operations $L_{i,j}(m)$ as the corresponding matrix obtained by adding m times row $j$ to row $i$.
\end{definition}
Example:\\
Let $\ell = 3$, $L_{1,2}(\alpha) = \begin{bmatrix}
    1&\alpha & 0\\0& 1 & 0\\ 0&0&1
    \end{bmatrix}$\\

From the cases with $\ell = 2$ and $\ell = 3$, we see that the weight of a function $wt(ev(f))$ for  $f\in \Fl$ depends on three aspects: the ambient matrix space $\HM$, the size of the maximal minor on $supp(f)$ and the spread of the maximal minor on $supp(f)$. We now introduce the following notation:
\begin{definition} Denote by 
$w_{n,k,s}$ the minimum weight of a function $ev(f)$ such that $f$ is obtained by evaluating $n \times n $ matrices with maximal minor size $k \times k$ and the smallest spread size among all $k \times k$ minors in $supp(f)$ is $s$.
\end{definition}
 Instead of inducting on the degree of the determinant functions, as in the case of affine Grassmann codes $\CA$, we will find the minimum distance of the code $\CH$ by inducting on $k$ and then $s$. Now we generalize some previous lemmas using the new notation.

\begin{lemma}\label{lem: reduce to spread}
Let $f \in \Fl$. Suppose that $f$ has a maximal minor of size $k$ whose spread has size $=s$. Then $wt(f) \geq q^{\ell^2-s^2}w_{s,k,s}$. %By a
\end{lemma}
\begin{proof}
    Let $f$ be as in the statement of the lemma. This implies there are $\ell-s$ rows and columns which do not appear in the spread of the maximal minor. Then for any of the $q^{\ell^2-s^2}$ values one can put on these columns, $f$ specializes to a combination of $s\times s$ determinants with the same maximal minors (though others may be changed due to the specialization). As each specialization has weight $w_{s,k,s}$ and there are $q^{\ell^2-s^2}$ specializations, the statement follows. 
\end{proof}

Note the following behavior in the $3\times 3$ case:\\
Let $\X = \begin{bmatrix}
    X_1&X_2& Y_1\\X_2^q&X_3 & Y_2\\ Y_1^q&Y_2^q&Y_3
    \end{bmatrix}$ and $f = det_{\{1,2\},\{2,3\}}(\X) =  \begin{vmatrix}
    X_2& Y_1\\X_3 & Y_2
    \end{vmatrix}$\\
Note that $f \in \Fl$ is a minor of size $2$ and spread $3$. Also recall that the following map $\X \mapsto A^{(q)T}\X A$ is an automorphism which preserves the weight of of a codeword, that is $wt(f(\X)) = wt(f(A^{(q)T} \X A))$. \\
Then let $A = \begin{bmatrix}
    1&0&1\\0&1&0\\ 0&0&1
    \end{bmatrix}$ the map $\X \mapsto A^{(q)T} \X A$ is as follows:\\
    $\begin{bmatrix}
    1&0&0\\0&1&0\\ 1&0&1
    \end{bmatrix} \begin{bmatrix}
    X_1&X_2& Y_1\\X_2^q&X_3 & Y_2\\ Y_1^q&Y_2^q&Y_3
    \end{bmatrix}  \begin{bmatrix}
    1&0&1\\0&1&0\\ 0&0&1
    \end{bmatrix} = \begin{bmatrix}
    X_1&X_2& X_1+Y_1\\X_2^q&X_3 & X_2^q+Y_2\\ X_1+Y_1^q& X_2+Y_2^q&X_1+Y_1+Y_1^q+Y_3
    \end{bmatrix}$\\
Then $$f =  \begin{vmatrix}
    X_2& X_1+Y_1\\X_3 & X_2^q+Y_2
    \end{vmatrix} = \begin{vmatrix}
    X_2& Y_1\\X_3 & Y_2
    \end{vmatrix}+\begin{vmatrix}
    X_2& X_1\\X_3 & X_2^q
    \end{vmatrix}=det_{\{1,2\},\{2,3\}}(\X)-det_{\{1,2\},\{1,2\}}(\X)$$
Note that applying this map, $f(A^{(q)T} \X A)$ is combination of minors in $\Fl$ containing a maximal minor of size $2$ and spread $2$. We may transform a function of a given size and spread to another function of the same weight with the same size but with smaller spread. 
The matrix $A$  we used in the transformation may be obtained from elementary matrices. In this example we used $L_{1,3}(1)$. In a more general sense, we state the following:\\
\begin{lemma}\label{lem:spread reduction}
Let $f \in \Fl$. Suppose that $f$ has a maximal minor of minimal spread size, $\mathcal{M}$ of size $k$ and spread $=s>k$. Then there exists $\bar{f} \in \Fl$ of the same weight with a maximal minor with size $k$ and spread size $\leq s-1$.
\end{lemma}

%WE ALSO NEED TO ESTABLISH THAT ALL OTHER MAXIMAL MINORS OF  F AND F BAR ARE THE SAME.
\begin{proof}
Let $\mathcal{M} = det_{I,J}(\X)$ be a maximal minor of size $k$ and spread size $s>k$ in the support of $f$.
Applying a suitable permutation of rows and columns, Without loss of generality, we may assume $$I = [k] \makebox{ and } J = \{s-k+1, s-k+2, ..., s\}.$$ 
%$s>k$ implies $|J-I| \neq \emptyset$\\
%Note $\exists A \in GL_{\ell}(\F_{q^2})$ such that $\X \mapsto A^{(q)T}\X A$ is equivalent to adding column $1$ to column $s$ and row $1$ to row $s$.\\
Let $\lambda \in \F_{q^2}^*$ such that $$-\lambda f_{I,J} - \lambda^qf_{I\cup \{s\}-\{1\},J\cup \{1\}-\{s\}} \neq 0.$$ We take $L_{1,s}(\lambda)\in GL_{\ell}(\F_{q^2})$. Consider the generic matrix $\Y=L_{1,s}(\lambda^q)\X L_{s,1}(\lambda)$. We shall prove that the codeword given by $\bar{f} = f(\Y) = f(L_{1,s}(\lambda^q)\X L_{s,1}(\lambda))$ has a maximal determinant with a smaller support. The multilinearity of the determinant implies that $$det_{I,J}(\Y) = det_{I,J}(\X) - det_{I,J\cup \{1\}-\{s\}}(\X).$$

Note that the only minors of the form $det_{A,B}(\Y)$ such that $det_{I,J\cup \{1\}-\{s\}}(\X)$ may appear in are the following:
$$\mathcal{Q} = det_{I\cup \{s\}-\{1\},J\cup \{1\}-\{s\}}(\Y), \mathcal{P} = det_{I\cup \{s\}-\{1\},J}(\Y).$$

Note the spread of $\mathcal{P}$ has size $s-1$, which contradicts that $\mathcal{M}$ is of minimal spread size. This implies $f_{I\cup \{s\}-\{1\},J} = 0$. This implies we only need to worry about $\mathcal{Q}$ and its corresponding $f_{I\cup \{s\}-\{1\},J\cup \{1\}-\{s\}}$. Thus the coefficient of $det_{I,J\cup \{1\}-\{s\}}(\X)$ in $\Bar{f}$ is  $-\lambda f_{I,J} - \lambda^qf_{I\cup \{s\}-\{1\},J\cup \{1\}-\{s\}} \neq 0$. Therefore $\mathcal{N} = det_{I,J\cup \{1\}-\{s\}}(\X) \in supp(\Bar{f})$ where $\mathcal{N}$ is size k and has spread $s-1$.

%we have a linear combination of determinants that includes $det_{I,I}(X)$.\\

\end{proof}
\begin{comment}
\begin{lemma}
Let $f \in \Fl$. Suppose that $f$ has a maximal minor of minimal spread, $\mathcal{M}$ of size $k$ and spread $=s=k+1$. Then we may find $\bar{f}$ of the same weight with a maximal minor of minimal spread, $\mathcal{N}$ with size $k$ and spread $s$.
\end{lemma}
\begin{proof}
Let $\mathcal{M} = det_{I,J}(\X)$ be a maximal minor of size $k$ and spread $s=k+1$ in the support of $f$.\\
Without loss of generality, we may assume $I = [k]$ and $J = \{2,3,...,k+1\}$\\
%$s>k$ implies $|J-I| \neq \emptyset$\\
Let $\lambda \in \F_{q^2}^*$ such that $\lambda*f_{I,J} + \lambda^q*f_{J,I} \neq 0$.\\
%Let $L_{1,s}(1) \in GL_{\ell}(\F_{q^2})$ such that $\X \mapsto A^{(q)T}\X A$ is equivalent to adding a $\lambda$*column $1$ to column $s$ and $\lambda^q$*row $1$ to row $s$.\\
We take $L_{1,s}(\lambda)\in GL_{\ell}(\F_{q^2})$
Let $\Y=A^{(q)T}\X A$, by the multilinearity of the determinant, this implies that $det_{I,J}(\Y) = det_{I,J}(\X) + \lambda*det_{I,I}(\X)$ and $det_{J,I}(\Y) = det_{J,I}(\X) + \lambda^q*det_{I,I}(\X)$.\\
By algebraic manipulation, this implies the new coefficient of $det_{I,I}(\X)$ in $\Bar{f}$ is $\lambda*f_{I,J} + \lambda^q*f_{J,I} \neq 0$\\
This implies $\mathcal{N} = det_{I,I}(\X) \in supp(\Bar{f})$ where $\mathcal{N}$ is size k and has spread $k$.\\
\end{proof}

As a direct consequence of the previous two lemmas, we obtain the following:
\end{comment}
As a direct consequence of the previous lemma, we obtain the following corollaries:
\begin{corollary}\label{cor:spread s to k}
Let $f \in \Fl$. Suppose that $f$ has a maximal minor of minimal spread, $\mathcal{M}$ of size $k$ and spread $=s>k$. Then we may find $\bar{f}$ of the same weight with a maximal minor of minimal spread, $\mathcal{N}$ with size $k$ and spread $k$.
\end{corollary}
\begin{proof} Simply apply Lemma \ref{lem:spread reduction} repeatedly.
\end{proof}
\begin{corollary}\label{cor:weight bound}
Let $f \in \Fl$. Suppose that $f$ has a maximal minor of size $k$ whose spread has size $=s$. Then $wt(f) \geq q^{\ell^2-k^2}(w_{k,k,k})$.
\end{corollary}
\begin{proof}
Let $f$ have a maximal minor of size $k$ whose spread has size $=s$. By Corollary \ref{cor:spread s to k}, we may apply an automorphism such that we transform $f$ to a function of the same weight with a maximal minor of minimal spread, with size $k$ and spread $k$. By Lemma \ref{lem: reduce to spread}, we may then state that  $wt(f) \geq q^{\ell^2-k^2}(w_{k,k,k})$. \end{proof}

\subsection{Finding $w_{n,k}$ via induction}
Having determined $w_{2,2}$, $w_{3,2}$ and $w_{3,3}$ (the base cases for $2\leq \ell \leq 3$), we shall now calculate $w_{\ell,k}$ for general $\ell$. Note that by Corollary \ref{cor:spread s to k}, we may assume that $f$ has a maximal minor of size $k$ whose spread is $k$. This implies we may assume we have a principal $k\times k$ minor in its support.

Therefore we shall define $w_{n,k}$.
\begin{definition}
$w_{n,k}$ denotes the minimum weight of a function $f \in \Fl$ evaluated on  $n \times n$ matrices, where $supp(f)$ has  a maximal minor of size $k \times k$.
\end{definition}
%We shall denote $w_{n,k}$ as the weight of a function $f$ such that $\ell = n$ and the maximal minor $\mathcal{M}$ is of size $k\times k$.\\
%For example: Let $n = 3$ and $k = 2$. This implies we are working over $3x3$ matrices with a maximal $2x2$ non principal minor.\\

\begin{lemma}\label{lem: induction}
$w_{k,k}\geq q^{k^2}-q^{k^2-1}-q^{k^2-3}+q^{k^2-2k}-1$
\end{lemma}
\begin{proof}
In Lemma \ref{lem:22classfier} and Lemma \ref{lem:33classifier}, the base cases for $k = 2$ and $k = 3$ are established. We assume the statement of the lemma is true for $2\leq k \leq K$ and prove that the statement holds for $K+1$. Without loss of generality we may assume $f$ has a maximal $K+1 \times K+1$ principal minor in its support.

As we did in the argument for the base cases, we specialize $det_{[K+1],[K+1]}(\X)$ along the $(K+1)-th$ column. Note there are exactly $q^{2(K+1)-1}-q^{2(K+1)-2}$ values for the $(K+1)-th$ column such that $x_{K+1,K+1} \neq -f_{[K],[K]}$. This leaves us with a non-trivial combination in the $K$ case and with a $K \times K$ minor in its support. Therefore, there are $q^{2(K+1)-1}-q^{2(K+1)-2}$ values for the specialization of the $(K+1)-th$ row and column where we specialize into the case $w_{K,K}$. These specializations contribute $(q^{2(K+1)-1}-q^{2(K+1)-2})w_{K,K}$ to $w_{K+1, K+1}$

We shall now consider the specialization where we do not obtain such a $K \times K$ maximal minor. This is the exclusive case where $x_{K+1,K+1} = -f_{[K],[K]}$. Hence we may assume there is no $K\times K$ minor in its support. Now we consider the possibilities for the $K-1 \times K-1$ minors. 

Note that all such minors are of the form  $det_{[K+1]-\{i,K+1\},[K+1]-\{j,K+1\}}(\X)$. Note that if the minor given by ${[K+1]-\{i,K+1\},[K+1]-\{j,K+1\}}$ does not appear in the partial evaluation, then the coefficients of the partial evaluation must satisfy the equation  $$f_{[K+1]-\{i,K+1\},[K+1]-\{j,K+1\}} = x_{i,K+1} {x_{j,K+1}}^q.$$ Lemma \ref{lem:systemsols} implies that the system of polynomial equations as stated above has at most $q+1$ solutions. 
%Note that the system such that all of the $K-1 \times K-1$ minors vanish is of the form stated in Lemma 8. 
Thus there are at least $q^{2(K+1)-2}-q-1$ values for the partial evaluation on the $(K+1)-th$ column such that we have a non-trivial combination with a $K-1 \times K-1$ minor in its support. These specializations contribute $q^{2K-1}(w_{K-1,K-1})$ to $w_{K+1, K+1}$. We put together all the inequalities and we obtain:

$$w_{K+1, K+1} \geq $$ 
\begin{tabular}{ll}

     & $(q^{2(K+1)-1}-q^{2(K+1)-2})w_{K,K}+ (q^{2(K+1)-2}-q-1)q^{2(K+1)-3}w_{K-1,K-1}$ \\
% \end{tabular}
% 
% \begin{tabular}{rcl}
%     & Which implies   &\\
    $=$   & $(q^{2K+1}-q^{2K})w_{K,K}+ (q^{2K}-q-1)(q^{2K-1})w_{K-1,K-1}$ \\
%\end{tabular}

%& By substituting our induction hypothesis &\\  

%\begin{tabular}{rcl}       
 $\geq$ &$(q^{2K+1}-q^{2K})(q^{K^2}-q^{K^2-1}-q^{K^2-3}+q^{K^2-2K}-1)+$ \\
  &  $(q^{2K}-q-1)(q^{2K-1})(q^{K^2-2K+1}-q^{K^2-2K}-q^{K^2-2K-2}+q^{K^2-4K+2}-1)$\\
%\end{tabular}
%& Through some algebraic manipulation, we obtain the following &\\ 
%\begin{tabular}{rcl}
  $\geq$ &$q^{K^2+2K+1}-q^{K^2+2K}-q^{K^2+2K-2}+q^{K^2+1}-q^{K^2} +q^{K^2-1}$\\
  & $+q^{K^2-2}+q^{K^2-3}-q^{K^2-2K+2}-q^{4K-1}-q^{2K+1}+2q^{2K}-q^{2K-1}$\\

%Note $q^{K^2+1}-q^{K^2}+q^{K^2-1}+q^{K^2-2}+q^{K^2-3}-q^{K^2-2K+2}-q^{4K-1}-q^{2K+1}+2q^{2K}-q^{2K-1} \geq q^{K^2-1}-1$, this implies\\
%& Note the following terms are a polynomial with a positive leading term which is strictly positive for $q\geq 2$. This implies &\\

  $\geq$ &$q^{K^2+2K+1}-q^{K^2+2K}-q^{K^2+2K-2}+q^{K^2-1}-1$\\

%& Equivalently &\\
%$w_{K+1,K+1}$ 
$\geq$ &$q^{{(K+1)}^2}-q^{{(K+1)}^2-1}-q^{{(K+1)}^2-3}+q^{{(K+1)^2-{2(K+1)}}}-1$\\

\end{tabular}\\
Therefore, by the principle of strong mathematical induction, the bound is met.
%$\Rightarrow w_{K+1,K+1} = $\\

%$\Rightarrow w_{K+1,K+1} \geq q^{K^2+2K+1}-q^{K^2+2K}-q^{K^2+2K-2}+q^{K^2+1}-q^{K^2}+q^{K^2-1}+q^{K^2-2}+q^{K^2-3}-q^{K^2-2K+2}-q^{4K-1}-q^{2K+1}+2q^{2K}-q^{2K-1}$\\

%$\Rightarrow w_{K+1,K+1} \geq q^{K^2+2K+1}-q^{K^2+2K}-q^{K^2+2K-2}+q^{K^2-1}-1$\\

%$\Rightarrow w_{K+1,K+1} \geq q^{{(K+1)}^2}-q^{{(K+1)}^2-1}-q^{{(K+1)}^2-3}+q^{{(K+1)^2-{2(K+1)}}}-1$
%By a similar proof Lemma 20, we may assume we are in the case where 
\end{proof}

\begin{proposition}\label{prop: mindist}
$w_{k,k}\geq q^{k^2}-q^{k^2-1}-q^{k^2-3}$
\end{proposition}
\begin{proof}
By Lemma \ref{lem: induction}, we have $w_{k,k}\geq q^{k^2}-q^{k^2-1}-q^{k^2-3}+q^{k^2-2k}-1$.
Note that for $k\geq 2$ we have $ q^{k^2}-q^{k^2-1}-q^{k^2-3}+q^{k^2-2k}-1 \geq q^{k^2}-q^{k^2-1}-q^{k^2-3}$. In fact, the equality is met for $k = 2$.
Therefore:
$$w_{k,k}\geq q^{k^2}-q^{k^2-1}-q^{k^2-3}$$
\end{proof}

Now we are finally ready to prove the main result of this paper
\begin{theorem}
Suppose that $\ell \geq 2$. Then $$d(\CH)  = q^{\ell^2}- q^{\ell^2-1}- q^{\ell^2-3}.$$ 
\end{theorem}
\begin{proof}
Assume $\ell \geq 2$. Let $f \in \Fl$ with a maximal minor of size $k$ and spread $s$. 
By Corollary \ref{cor:weight bound}, $wt(f) \geq q^{\ell^2-k^2}(w_{k,k,k})$. By Proposition \ref{prop: mindist} $$w_{k,k,k}\geq q^{k^2}-q^{k^2-1}-q^{k^2-3}.$$ Taken together  $$wt(f) \geq q^{\ell^2-k^2}(q^{k^2}-q^{k^2-1}-q^{k^2-3}) = q^{\ell^2}-q^{\ell^2-1}-q^{\ell^2-3}.$$ 
This implies $d(\CH) \geq q^{\ell^2}- q^{\ell^2-1}- q^{\ell^2-3}.$

Equality is obtained by $$f= det_{\{1,2\},\{1,2\}}(\X)+1.$$ % then $wt(f) = q^{\ell^2}- q^{\ell^2-1}- q^{\ell^2-3}$. Therefore $$d(\CH)  = q^{\ell^2}- q^{\ell^2-1}- q^{\ell^2-3}.$$ 
\end{proof}

\begin{corollary}
Let $\ell \geq 2$, $d(\CH) \geq d(\CA)$. That is $$q^{\ell^2}- q^{\ell^2-1}- q^{\ell^2-3} \geq  \prod_{i=0}^{\ell-1} (q^{\ell}-q^{i}) $$
\end{corollary}
\begin{proof}

Note that $q^{\ell^2}- q^{\ell^2-1}- q^{\ell^2-3} \geq (q^2-1)(q^2-q)q^{\ell^2-4} \geq  \prod_{i=0}^{\ell-1} (q^{\ell}-q^{i}) $.

\end{proof}

The most remarkable fact about the affine Hermitian Grassman code is that it a has better minimum distance than the affine Grassmann code while having the same length and dimension. In the tables below, we compare the parameters of these codes for $\ell = 2$ and $\ell = 3$.\\
\begin{tabular}{cc}
     
    \begin{tabular}{|c|c|c|c|c|}
        \hline
        %%%%For tables, you have to add the & to represent the cells
        q & n & k & $d(C^{\mathbb{A}}(2,4))$ & $d(C^{\mathbb{H}}(2))$\\
        \hline
        2 & 16 & 6 & 6 & 6\\
        \hline
        3 & 81 & 6 & 48 & 51\\
        \hline
        4 & 256 & 6 & 180 & 188\\
        \hline
        5 & 625 & 6 & 480 & 495\\
        \hline
        7 & 2,401 & 6 & 2,016 & 2,051\\
        \hline
        8 & 4,096 & 6 & 3,528 & 3,576\\
        \hline
        9 & 6,561 & 6 & 5,760 & 5,823\\
        \hline
        %11 & 14,641 & 6 & 13,200 & 13,299\\
        %\hline 
    \end{tabular}
    & 
    \begin{tabular}{|c|c|c|c|c|}
        \hline
        %%%%For tables, you have to add the & to represent the cells
        q & n & k & $d(C^{\mathbb{A}}(3,6))$ & $d(C^{\mathbb{H}}(3))$\\
        \hline
        2 & 512 & 20 & 168 & 192\\
        \hline
        3 & 19,683 & 20 & 11,232 & 12,393\\
        \hline
        4 & 262,144 & 20 & 181,440 & 192,512\\
        \hline
        5 & 1,953,125 & 20 & 1,488,000 & 1,546,875\\
        \hline
        7 & 40,353,607 & 20 & 33,784,128 & 34,471,157\\
        \hline
        8 & 134,217,728 & 20 & 115,379,712 & 117,178,368\\
        \hline
        9 & 387,420,489 & 20 & 339,655,680 & 343,842,327\\
     %   \hline
    %    11 & 2,357,947,691 & 20 & 2,124,276,000 & 2,141,817,249\\
        \hline 
    \end{tabular} \\
\end{tabular}
\\

\section{Dual Code}
In this section we study some properties of the dual code $\CH^\perp$ and its relation to the dual affine Grassmann code $\CA^\perp$. In particular, we prove that the minimum distance codewords of the dual codes are similar for both the dual of affine Grassmann codes and dual affine Hermitian Grassmann codes. We begin by defining a dual code. Please recall the following:
\begin{definition}
Let $C$ be an $[n,k]$ linear code, then its dual code $C^\perp$ is its orthogonal complement as a vector space. That is, $C^\perp$ is an $[n, n-k]$ code such that:

$$C^\perp := \{ h \in \F_q^n \ | \ \sum\limits_{i=1}^n c_ih_i = 0 \ \forall c \in C \}$$  
\end{definition}

Let $\CH^\perp$ denote the Dual of the affine Hermitian Grassmann code. Recall that $$\CH^\perp := \{ h \in \F_{q^2}^{\HM} \ | \ \sum\limits_{H\in \HM} h_H f(H) = 0 \ \forall f \in \Fl  \} $$

Then we know the following

$\CH^\perp$ is an $[q^{\ell^2}, q^{\ell^2}-\binom{2\ell}{\ell}]$ code and $\CH^\perp = {\CH^\perp}^(q)$. The latter fact is due to  \cite{Sticht}. 

In fact, we know more about the minimum distance of $\CA^\perp$.
\begin{proposition}\cite[Theorem 17]{BGT2}

Let $\ell \geq 2$. The minimum distance $d(\CA^\perp)$ of the code $\CA^\perp$ satisfies:
$$
d(\CA^\perp)=  \begin{cases} 
      3 & q>2 \\
      4 & q = 2 
   \end{cases}
$$
\end{proposition}

In subsequent work, one of the named authors along with P. Beelen characterized the minimum distance codewords of $\CA^\perp$.

\begin{definition}
Let $f \in \CH^\perp$ $supp(f) = \{H\in \mathbb{H} | c_{H} \neq 0 \}$
\end{definition}

\begin{definition}
Let $k\leq \ell$, we denote $I_k$ as the matrix with an $k\times k$ identity block and the remaining entries are 0. That is $a_{i,i} = 1$ if $1\leq i\leq k$.
\end{definition}

\begin{definition}
We denote $E_{i,j}$ to be the $\ell \times \ell$ matrix which all entries equal 0 except the $(i,j)-th$ entry which equals 1.
\end{definition}

\begin{proposition}
\cite[Theorem 8]{AfGrass}
Let $\ell \geq 2$, let $q > 2$ and let $c\in \CA^\perp$ be a weight 3 codeword with support $supp(c) = \{N_1, N_2, N_3\}$. Then there exists an automorphism such that we may map $c \rightarrow c'$ where $supp(c') = \{ 0, I_1, \alpha I_1\}$ and $\alpha = (\frac{c_{N_2}}{c_{N_1}+c_{N_2}})$.\\
Conversely, given $\alpha \in \F_{q}\setminus \{0,1\}$, there exists a codeword $c\in \CA^\perp$ with $supp(c) = \{ 0, I_1, \alpha I_1\}$. Its nonzero coordinates satisfy
$$ c_{I_1} = \frac{-\alpha}{\alpha -1}c_0  \makebox{ and } c_{\alpha I_1} = \frac{1}{\alpha -1}c_0$$
\end{proposition}

\begin{proposition}
\cite[Theorem 15]{AfGrass}
Let $\ell \geq 2$. Let $q= 2$ and let $c$ be a codeword of $\CA^\perp$ of weight $4$. Suppose that $supp(c) = \{M_1, M_2, M_3, M_4\}$. Then there exists an automorphism such that we may map $c \rightarrow c'$ where $supp(c')$ is one of the following:
\begin{enumerate}
    \item[i]  $\{ 0, E_{1,1}, E_{1,2}, I_1+E_{1,2}\}$
    \item[ii] $\{ 0, E_{1,1}, E_{2,1}, I_1+E_{2,1}\}$
    \item[iii] $\{ 0, E_{1,1}, E_{1,2}+E_{2,1}, E_{1,1}+E_{1,2}+E_{2,1}\}$
\end{enumerate}

\end{proposition}

%DOEL cita la caracterizacion de las palabars minimas.
\begin{definition}
Let $C$ be a linear code and $T$ be a set of coordinates in $C$. We define the puncturing of $C$ on $T$ as the resulting linear code $C^T$ from deleting all coordinates in $T$ in each codeword of $C$.
\end{definition}
\begin{definition}
Let $C$ be a linear code, $T$ be a set of coordinates in $C$ and $C(T)$ the set of codewords which are 0 on $T$. We define the shortening of $C$ as the puncturing of $C(T)$ on $T$.
\end{definition}
We remark that these code operations are duals of each other. That is, the dual code of  puncturing $C$ is shortening $C^\perp$.
%%DOEL, Define Puncturing y shortening. Escribe que ambos son duales uno del otro.

With the fact that $\CH$ is a puncturing of $\CA$ on the positions outside of $\HM$, and that $\CH^\perp$ is the linear code obtained by shortening $\CA^\perp$ at the positions outside of $\HM$, we shall determine the minimum distance of $\CH^\perp$. %then focus on its minimum distance and its minimum weight codewords. For this purpose, it will be useful to describe the codewords of $\CH^\perp$ by their support.\\

%First we complete the easy case, that is, determining $d(\CH^\perp)$ for $q > 2.$

\begin{theorem}
Let $\ell \geq 2$. The minimum distance $d(\CH^\perp)$ of the code $\CH^\perp$ satisfies:
$$
d(\CH^\perp)=  \begin{cases} 
      3 & q>2 \\
      4 & q = 2 
   \end{cases}
$$
\end{theorem}
\begin{proof}
Recall we are puncturing the code $\CA$ at the matrices in $\mathbb{M}^{\ell \times \ell}(\F_{q^2}) \setminus \HM$ to obtain $\CH$. This implies we are shortening the code $\CA^\perp$ to obtain $\CH^\perp$. By \cite{AfGrass}, this implies we have a lower bound$$
d(\CH^\perp)\geq  \begin{cases} 
      3 & q>2 \\
      4 & q = 2 
   \end{cases}.
$$\\

If $q>2$, there exists $c\in \CA^\perp$ such that $supp(c) = \{ 0, E_{1,1}, \alpha E_{1,1} \}$ and $\alpha \in \F_q, \alpha \neq 0,1$. Because all 3 matrices are Hermitian, this implies that when shortening we have a codeword $c' \in \CH^\perp$ such that $supp(c) = supp(c')$. Therefore, for $q>2$, $d(\CH^\perp)= 3$.

If $q=2$, suppose by way of contradiction that $\#supp(c) = 3$.
Let $A,B,C \in supp(c)$. This implies $c_A = c_B = c_C = 1$. This implies $c \cdot \Bar{1} = c_A+c_B+c_C = 1 \neq 0$
Note $\Bar{1} \in \CH$, which contradicts that $c \in \CH^\perp$. Therefore $\#supp(c) \neq 3$.

Note there exists $c\in \CA^\perp$ such that the support set consists of the matrices $supp(c) = \{ 0, E_{1,1}, E_{1,2}+E_{2,1}, E_{1,2}+E_{2,1}+E_{1,1}\}$. % and $\Bar{(E_{1,2} +E_{2,1})} = {(E_{1,2}+E_{2,1})}^T$.
Because all 4 matrices are Hermitian, this implies that when shortening we have a codeword $c' \in \CH^\perp$ such that $supp(c) = supp(c')$. Therefore, for $q=2$, $d(\CH^\perp)= 4$.\\
\end{proof}

As in the case for affine Grassmann codes, we characterize all minimum distance codewords of $\CH^\perp$. The case for $q > 2$ can be done simply by considering the shortening operation. As there are minimum distance codewords of $\CA^\perp$ over $q^2$ whose support is entirely of Hermitian matrices. It is these codewords which are the minimum distance codewords of $\CH^\perp$. However, when $q =2$, the corresponding code $\CA^\perp$ is actually a quaternary code of distance $3$. In this case, shortening from the affine Grassmannian into the affine Hermitian Grassmannian eliminates all codewords of weight $3$ from $\CH^\perp$ Interestingly, despite the code $\CH$ being quaternary, it behaves similar to the binary case for the affine Grassmann codes. Remarkably, we characterize all the minimum distance codewords of $\CH^\perp$ with the same technique used in \cite{AfGrass}.

\begin{lemma}\label{lem: translate q = 2}
Let $q = 2$, $\ell \geq 2$ and $c\in \CH^\perp$ and $supp(c) = \{A, B, C, D\}$. Then there exists an automorphism $\sigma \in Aut(\CH)$ such that $supp(\sigma(c)) = \{ M, N, P, 0\}$.
%$supp(\sigma(c)) = \{ 0, E_{1,1}, E_{1,2}+E_{2,1}, E_{1,1}+E_{1,2}+E_{2,1}\}$.
\end{lemma} 
\begin{proof}
Let $A,B,C,D\in supp(c)$. %As $\CH$ is a $2$--invariant code, its minimum distance codewords are multiples of binary vectors. Thus, without loss of generality, we may assume $c_A = c_B = c_C = c_D = 1$. 
Now we translate the support by $D$ by applying the automorphism $f(X) \mapsto f(X+D)$. Then, let $M = A+D$, $N = B+D$ and $P = C+D$, $supp(\sigma(c)) = \{ M, N, P, 0\}$ %Without loss of generality we assume $supp(c) = \{ A, B, C, 0 \}$.
\end{proof}

\begin{lemma}\label{lem: automorphism q=2}
Let $q = 2$, $\ell \geq 2$ and $c\in \CH^\perp$ and $supp(c) = \{A, B, C, 0\}$. Then there exists an automorphism $\sigma \in Aut(\CH)$ such that $supp(\sigma(c)) = \{ 0, I_{k}, E_{1,2}+E_{2,1}, I_{k}+E_{1,2}+E_{2,1}\}$ where $k = min\{ rank(A), rank(B), rank(C) \}$.
%$supp(\sigma(c)) = \{ 0, E_{1,1}, E_{1,2}+E_{2,1}, E_{1,1}+E_{1,2}+E_{2,1}\}$.
\end{lemma} 
\begin{proof} Let $A,B,C,0\in supp(c)$. As $\CH$ is a $2$--invariant code, its minimum distance codewords are multiples of binary vectors. Thus, without loss of generality, we may assume $c_A = c_B = c_C = c_0 = 1$. Without loss of generality, we may assume $A$ is of minimum nonzero rank. %Now we translate the support by $D$ by applying the automorphism $f(X) \mapsto f(X+D)$.  Without loss of generality we assume $supp(c) = \{ A, B, C, 0 \}$. % where $A = A_1-D_1$ $B = B_1-D_1$ and $C = C_1-D_1$.

Note that for $f \in \Fl$, $c \cdot ev(f) = 0$ implies $A^{I,J}+B^{I,J}+C^{I,J}+0^{I,J} = 0$ for all minors $det_{I,J}(\X)$. As this also holds for all $1 \times 1$ minors we obtain $C= A+B$. 

Without loss of generality, we may assume $C_{1,1} = 1$ with $A_{1,1} = 1$ and $B_{1,1} = 0$. %As $A$ has its $(1,1)$ entry equal to $1$, we may use column operations on the right to eliminate all other entries in the first row. Likewise, the conjugate operations on the left, may be used to eliminate all other entries in the first column. We may assume that $A_{1,i} = A_{i,1} = 0$ for $i>1$.
Because $A$ is Hermitian, by \cite{fulton_1977}, there exists $P \in GL_{\ell}(\F_{q^2})$ such that ${{P^{(q)}}^T}AP = I_k$. By applying this automorphism $\sigma_1$, $supp(\sigma_1(c)) = \{ A', B', C', 0\}$ where $C'_{1,1} = 1$ with $A'_{1,1} = 1$ and $B'_{1,1} = 0$. Note $rnk(A) \geq 1$ implies $A_{1,i} = A_{i,1} = 0$ for $i>1$.

Now that we have established that $A$ has a special form, we will use the condition $c \cdot ev(f) = 0$ to determine the entries of $A$  and $B$. 

First let $f = det(\{1,i\},\{1,j\})$ with $i\geq 2$ and $j \geq 2$.
Note:
$$f(A') = A'_{1,1}A'_{i,j}-A'_{1,j}A'_{i,1} = A'_{i,j},$$
$$f(B') = B'_{1,1}B'_{i,j}-B'_{1,j}B_{i,1} = B'_{1,j}B'_{i,1},$$
$$f(C') = C'_{1,1}C'_{i,j}-C'_{1,j}C_{i,1} = (A'_{i,j}+B'_{i,j})+(B'_{1,j}B'_{i,1})$$

%$$f(A) = A_{1,1}A_{i,j}-A_{1,j}A_{i,1} = A_{i,j},
%f(B) = B_{1,1}B_{i,j}-B_{1,j}B_{i,1} = B_{1,j}B_{i,1}, f(C) = C_{1,1}C_{i,j}-C_{1,j}C_{i,1} = (A_{i,j})+(B_{1,j}B_{i,1})$$\\

Then $c \cdot ev(f) = 0$ implies $$A'_{i,j}+B'_{1,j}B'_{i,1}+(A'_{i,j}+B'_{i,j})+(B'_{1,j}B'_{i,1}) = 0.$$
This implies $B'_{i,j} = 0$ for all $i,j$ such that $i\geq 2$ or $j\geq 2$.

%In a similar manner, taking $f = det(\{1,2\},\{1,2\})$, we conclude $B_{2,2} = 0$.\\

Because $B'\neq 0$, there must be a nonzero entry in row or column 1. Without loss of generality after permuting rows or columns we may assume this is $B'_{1,2}$. Recall that because $B'$ is Hermitian, $B'_{2,1} = {B'_{1,2}}^2 \neq 0.$ Without loss of generality we may assume $B'_{1,2} = B'_{2,1} = 1$. This is because $x^3 = 1$ for all nonzero values in $\F_4$. Thus, automorphisms defined by multiplying a nonzero scalar to a row and its conjugate to a column do not affect $I_k$. %As in the case with the matrix $A$, 
Then, through proper elementary row operations, we may use the entry in the second column of $B'$, $B'_{1,2} = 1$ to eliminate all other entries in row $1$ and column $1$. As $A'_{1,2} = A'_{2,1} = 0$ these operations do not change $A'$. Thus we may assume $B' = E_{1,2}+E_{2,1}$. %Then, let $rnk(A) = k$ we may perform elementary row operations to reduce $A$ to $I_k$ without altering $B$.
Therefore $supp(\sigma(c)) = \{ 0, I_{k}, E_{1,2}+E_{2,1}, I_{k}+E_{1,2}+E_{2,1}\}$. %Therefore, $supp(c') = \{ 0, E_{1,1}, E_{1,2}+E_{2,1}, E_{1,1}+E_{1,2}+E_{2,1}\}.$ 
\end{proof}
We remark that since all automorphism were from multiplying by invertible matrices, then $k = min\{ rnk(I_k), rnk(E_{1,2}+E_{2,1}), rnk(I_{k}+E_{1,2}+E_{2,1})\}$. Moreover, this implies $k\leq 2 = rnk(E_{1,2}+E_{2,1})$

\begin{lemma}\label{lem: k =1}
Let $q = 2$, $\ell \geq 2$ and $c\in \CH^\perp$ and $$supp(c) = \{0, I_{k}, E_{1,2}+E_{2,1}, I_{k}+E_{1,2}+E_{2,1}\}$$ where $$k = min\{rnk(I_{k}), rnk(E_{1,2}+E_{2,1}), rnk(I_{k}+E_{1,2}+E_{2,1}) \}.$$ Then $k = 1$.
%$supp(\sigma(c)) = \{ 0, E_{1,1}, E_{1,2}+E_{2,1}, E_{1,1}+E_{1,2}+E_{2,1}\}$.
\end{lemma}
\begin{proof}
Let $c$ be as stated above. Recall $k\leq 2$, thus we need only show $k\neq 2$.  %Then by the previous Lemma, this implies there exists $\sigma_1 \in Aut(\CH)$ such that $supp(\sigma_1(c)) = \{ 0, I_{k}, E_{1,2}+E_{2,1}, I_{k}+E_{1,2}+E_{2,1}\}$ where k is the rank of $A$. \\
Let's assume that $k= 2$. This implies that $I_{k}+E_{1,2}+E_{2,1} = I_{2}+E_{1,2}+E_{2,1}$. This implies $rnk(I_{k}+E_{1,2}+E_{2,1}) = 1$, which contradicts the previous remark that $k = min\{ rnk(I_k), rnk(E_{1,2}+E_{2,1}), rnk(I_{k}+E_{1,2}+E_{2,1})\}$. Therefore, $k = 1$.
\end{proof}

\begin{lemma}\label{lem: supp char}
Let $q = 2$, $\ell \geq 2$ and $c\in \CH^\perp$ and $supp(c) = \{A, B, C, D\}$. Then there exists an automorphism $\sigma \in Aut(\CH)$ such that $supp(\sigma(c)) = \{ 0, I_{1}, E_{1,2}+E_{2,1}, I_{1}+E_{1,2}+E_{2,1}\}$.
%$supp(\sigma(c)) = \{ 0, E_{1,1}, E_{1,2}+E_{2,1}, E_{1,1}+E_{1,2}+E_{2,1}\}$.
\end{lemma}
\begin{proof}
Let $c$ be as stated above, by Lemma \ref{lem: translate q = 2} and Lemma \ref{lem: automorphism q=2} , there exists an automorphism $\sigma \in Aut(\CH)$ such that $supp(\sigma(c)) = \{ 0, I_{k}, E_{1,2}+E_{2,1}, I_{k}+E_{1,2}+E_{2,1}\}$. By Lemma \ref{lem: k =1}, $k = 1$. Therefore, $supp(\sigma(c)) = \{ 0, I_{1}, E_{1,2}+E_{2,1}, I_{1}+E_{1,2}+E_{2,1}\}$
\end{proof}

%As a direct consequence, we obtain the following:
%\begin{corollary}
%If $q = 2$ then $d(\CH^\perp) = 4$.

%\end{corollary}
\begin{lemma}
Let $q = 2$. The support of a minimum weight codeword of $\CH^\perp$ is contained in the class:
$$\{ H, H+{a_1^*a_1}, H+{a_2^*a_1+a_1^*a_2}, H+{a_1^*a_1}+{a_2^*a_1+a_1^*a_2} \}$$
Where $H \in \HM$, meanwhile $a_1, a_2$ are linearly independent vectors $\in {\F_{4}}^{\ell}$.
\end{lemma}
\begin{proof}
By Lemma \ref{lem: supp char}, we may apply an automorphism of $\CH^\perp$ to map any of the possible support sets $\{A,B,C,D\} = supp(c)$ to a codeword with support $\{ 0, I_{1}, E_{1,2}+E_{2,1}, I_{1}+E_{1,2}+E_{2,1}\}$.
By acting on $\{ 0, I_{1}, E_{1,2}+E_{2,1}, I_{1}+E_{1,2}+E_{2,1}\}$ with an invertible matrix $A$ satisfying $e_1 A = a_1, e_2A = a_2$ via the action $M \mapsto A^*MA$, we map the set $\{ 0, I_{1}, E_{1,2}+E_{2,1}, I_{1}+E_{1,2}+E_{2,1}\}$ to the set $\{A^*0A, A^*I_1A, A*(E_{1,2}+E_{2,1})A, A*(I_{1}+E_{1,2}+E_{2,1})A \}$ Which is equal to the set $\{ 0, {a_1^*a_1}, {a_2^*a_1+a_1^*a_2}, {a_1^*a_1}+{a_2^*a_1+a_1^*a_2} \}$ where $a_1$ is the first row $A$ and $a_2$ is the second row of $A$. Similarly, $a_1^*$ is the first column $A^*$ and $a_2^*$ is the second column of $A^*$. Finally, we act by adding $H \in \HM$, thus obtaining
$$\{ H, H+{a_1^*a_1}, H+{a_2^*a_1+a_1^*a_2}, H+{a_1^*a_1}+{a_2^*a_1+a_1^*a_2} \}$$
\end{proof}

\section{Conclusion}

In this paper we have introduced the affine Hermitian Grassmann codes. These are linear codes associated to the affine part of the polar Hermitian Grassmannian, defined in the same way affine Grassmann codes are defined from the Grassmannian. As might be expected, the affine Hermitian Grassmann code is very similar to the affine Grassmann code. The affine Hermitian Grassmann code for $q$ is defined over the field $\F_{q^2}$ and obtained from puncturing $\CA$ as a code over $\F_{q^2}$. However, $\CH$ is a $q$--invariant code. This implies $\CH$ has a basis over $\F_q$. Furthermore it has the same dimension and better minimum distance than the corresponding affine Grassmann code over $\F_q$. The similarities are deeper, as the dual code $\CA^\perp$ is similar to $\CA^\perp$. Even the support of the minimum weight codewords are nearly identical in a case by case basis. We hope this remarkable similarity holds for other codes related to the Grassmannian.

\section{Acknowledgments}
The authors are very thankful to Sudhir Ghorpade for his insightful comments on the automorphism group and his suggestions which have greatly improved this manuscript. This research is supported by NSF-DMS REU 1852171: REU Site: Combinatorics, Probability, and Algebraic Coding Theory and NSF-HRD 2008186: Louis Stokes STEM Pathways and Research Alliance: Puerto Rico-LSAMP - Expanding Opportunities for Underrepresented College Students (2020-2025)

\begin{comment}
\nocite{ellis-monaghan_tutte_nodate}
\nocite{ellis-monaghan_distance_2007}
\nocite{eubanks-turner_interlace_2018}
\nocite{arratia_interlace_2000}
\nocite{danielsen_graphs_2013}
\nocite{danielsen_interlace_2010}
\nocite{danielsen_edge_2007}
\nocite{morse_interlace_nodate}
\nocite{ballister_alternating_nodate}
\nocite{schlingemann_quantum_2001}
\nocite{grassl_graphs_2002}
\nocite{aigner_course_2007}
\nocite{bollobas_evaluations_2002}
\nocite{ellis-monaghan_new_1998}
\nocite{hobbs_william_nodate}
\nocite{ellis-monaghan_deletion-contraction_2019}
\nocite{ellis-monaghan_combinatorial_2019}
\nocite{las_vergnas_polynome_1983}
\nocite{martin_enumerations_nodate}
\nocite{aigner_interlace_2002}
\nocite{goodall_tutte_2014}
\nocite{austin_circuit_2007}
\nocite{brijder_interlace_2014}
\nocite{10.1145/2422436.2422494}
\nocite{lopez2020hermitianlifted}
\end{comment}
\nocite{Sticht}
\nocite{BGT}
\nocite{AfGrass}

\bibliographystyle{plain}
\bibliography{references}
%\bibliography{final_graph_poly_bib}

%\bibliographystyle{unsrt}  

%%% Comment out this section when you \bibliography{references} is ena

\end{document}